\newif\ifcsl
\cslfalse

\ifcsl
\documentclass{lipics-v2021}
\nolinenumbers
\else
\documentclass[reqno]{amsart}
\usepackage[foot]{amsaddr}
\fi

\usepackage{proof-at-the-end}
\ifcsl
\pgfkeys{/prAtEnd/global custom defaults/.style={
    proof at the end,
    restate,
    no link to proof,
    text proof={\proofname\  (statement
      no. \string\ref{thm:prAtEnd\pratendcountercurrent})},
  }}
\else
\pgfkeys{/prAtEnd/global custom defaults/.style={
    normal,
  }}

\fi

\input{macros}

\ifcsl
\input{CSLfirstpage}
\fi
\makeatletter
\def\input@path{{sections/}{./}}
\makeatother

\begin{document}

\ifcsl

\maketitle
\begin{abstract}
  For a functor $Q$ from a category $\C$ to the category $\Pos$ of
ordered sets and order-preserving functions, we study liftings of
various kind of structures from the base category $\C$ to the total
(or \Gr) category $\iQ$.  That lifting a monoidal structure
corresponds to giving some lax natural transformation making $Q$
almost monoidal, might be part of folklore in category theory.
We rely on and generalize the tools supporting this correspondence so
to provide exact conditions for lifting symmetric monoidal closed and
star-autonomous structures.

A corollary of these characterizations is that, if $Q$ factors as a
monoidal functor through $\SLatt$, the category $\SLatt$ of complete
lattices and \sp functions, then $\iQ$ is always symmetric monoidal
closed. In this case, we also provide a method, based on the double
negation nucleus from quantale theory, to turn $\iQ$ into a
star-autonomous category.
The theory developed, originally motivated from the categories
$\QSet[P]$ of \SdP, yields a wide generalization of \HS construction
of star-autonomous categories by means of orthogonality structures.

\end{abstract}
\else
\makeatletter
\newcommand{\definetitlefootnote}[1]{
  \newcommand\addtitlefootnote{
    \makebox[0pt][l]{$^{*}$}
    \footnote{\protect\@titlefootnotetext}
  }
  \newcommand\@titlefootnotetext{\spaceskip=\z@skip $^{*}$#1}
}
\makeatother

\definetitlefootnote{
  Work supported by the ANR projects LAMBDACOMB ANR-21-CE48-0017 and
  RECIPROG ANR-21-CE48-0019}

\title[Lifting star-autonomous structures]{Lifting star-autonomous structures\addtitlefootnote}
\author[C. De Lacroix]{Cédric De Lacroix}
\email{cedric.delacroix@lis-lab.fr}
\author[G. Chichery]{Gregory Chichery}
\email{gregory.chichery@lis-lab.fr}
\author[L. Santocanale]{Luigi Santocanale}
\email{luigi.santocanale@lis-lab.fr}
\address{
  LIS, CNRS UMR 7020, Aix-Marseille Universit\'e,
  France
}

\maketitle

\begin{abstract}
  
\end{abstract}

\sautcat, dualizing object, Girard quantale, double
negation nucleus.

\fi

\section{Introduction}

Categorical models of proofs of multiplicative classical linear logic
are \sautcats \cite{Barr1979}.  Many of these categories are built
from a given \sautcat, usually a degenerate one such as the category
of sets and relations, by attaching to each object a structure and by
requiring the maps to be compatible with the structures.  The new
category is usually better behaved, at least, for the semantics of
proofs.  Most often, the collection of all the structures that we can
attach to an object is a poset. This is the case, for example, of \SdP
categories $\QFSet[P]$
\cite{SchalkDePaiva2004} and of \HS orthogonality \cats
\cite{HS2003,Fiore2023}.

From a categorical perspective, attaching to each object a structure
amounts to considering the total category (or \Gr category) $\iQ$ of a
(monoidal) functor $Q : \C \rto \Pos$, where $\C$ is some \SM category
and $\Pos$ is the category of posets and \opr functions. This
construction yields 
the canonical (op-)fibration $\pi : \iQ \rto \C$ which strictly
preserves the monoidal structure.  The theory of monoidal fibrations
was firstly developed in \cite{Shulman2008} and later, from a
different perspective more relevant here, in \cite{MV2020}.  Roughly
speaking, we investigate in this paper variants of the following
questions: given a monoidal functor $Q : \C \rto \Pos$ (so that we
are ensured that the total category $\iQ$ is monoidal), when is $\iQ$
a 
closed \cat, and when is it \saut?

\medskip

More precisely, we give exact answers (and characterizations) to the
following questions: given a monoidal functor $Q : \C \rto \Pos$,
where $\C$ is monoidal with some additional structure, when has $\iQ$
this structure in a such way that the canonical op-fibration $\pi$
strictly preserves it?  We call this the lifting problem for a
structure. The structures considered here are being closed, being
\saut, and having \final coalgebras (and initial algebras) of
functors. Yet the tools developed in this paper can, in principle, be
used to lift other kind of structures, if not all the structures.

When the monoidal functor $Q$ takes values in the category \SLatt of
complete lattices and \sp functions, our characterizations yield
remarkable consequences. In this case, $\iQ$ turns out to always be
closed. Moreover, assuming that $\C$ is \saut with dualizing object
$0$, it is possible to turn $\iQ$ into a \sautcat by choosing an
element $\omega \in Q(0)$, which might be thought of as a sort of
global falsity, and by considering a double negation quotient of
$\iQ$.  While this construction yields a generalization of focused
orthogonality categories \cite{HS2003,Fiore2023}, the way we
discovered it was through the analogy with quantales, that is,
provability models of intuitioninistic linear logic, see
\cite{Rosenthal1990b}.  In order to turn a quantale into a model of
classical linear logic, a so called Girard quantale, see
\cite{Rosenthal1990a}, it suffices to choose a candidate falsity (a
candidate dualizing element) and consider the fixed points of the
double negation nucleus it gives rise.

This research was partly motivated by recent research on the algebraic
and categorical semantics of linear logic with fixed points
\cite{EJ2021,AdJS2022,Farzad2023,Fiore2023} extending to linear logic
previous work by one of authors on the categorical semantics of
fixed-point logics and circular proof systems
\cite{Santocanale2002,Santocanale02a,Santocanale02b,FSCSAL13}.  An
important model of proofs of linear logic with fixed points, that has
been considered in those works, is the category $\Nuts$ of non uniform
totality spaces, which indeed arises as $\iQ$ for a functor
$Q : \Rel \rto \Pos$.  
Besides 
considering models of proofs of linear logic, many other reasons have
triggered us to this research. Persuaded that ordered structures are
pervasive and essential both in logic and computation, we started
investigating Frobenius quantales as Frobenius monoids in the \sautcat
$\SLatt$ and, later, Frobenius monoids in arbitrary
\sautcats. 
We could prove in \cite{delacroix_et_al:LIPIcs.CSL.2023.18} the
equivalence between an object $X$ being nuclear and the monoid
$X\impl X$ being Frobenius under the hypothesis that the tensor unit
embeds into $X$ as a retract. To argue that this hypothesis is
necessary, we built a counterexample by resorting to \SdP categories
$\QSet[P]$ \cite{SchalkDePaiva2004} for a well-chosen Girard quantale
$P$.  These categories, also of the form $\iQ$ for a functor
$Q : \Rel \rto \Pos$, turned out to be extremely interesting for other
reasons.
For example---also considering the possible generalizations of these
categories that we hint to in this paper---they can accommodate in a
uniform way many categories of fuzzy sets and/or relations that have
been considered in the literature, see
e.g. \cite{Winter2007,HWW2014,Mockor2015}. These categories rely on
the unit interval $[0,1]$ and on some of its quantale structures (the
Łukasiewicz quantale, its Heyting algebra structure, the Lawvere
quantale, see e.g. \cite{Bacci2023}).  The categories of the form
$\QSet[P]$ deeply exemplify an interplay between quantales, that is,
provability models of linear logic, and \SMCC, the models of proofs of
the same logic.
For intuitionistic logic, this interplay is well known, the
connections between Heyting algebras, Cartesian closed categories, and
topoi being the object of several monographs.  Whether this interplay
is still relevant for linear logic is, in our opinion, unclear. The
results presented in this paper can be understood as providing
evidence and suggesting a positive answer to this question.

  \medskip
  
  This paper is organised as follows. 
  We give in Section~\ref{sec:background} some background on posets,
  complete lattices, and monoidal categories; by doing so, we also
  settle the notation. We also define in this section the total
  category $\iQ$ and recall some of its properties.
  In Section~\ref{sec:liftingFunctors} we characterise how to lift a
  functor of several variables, some contravariants and some
  covariants, from $\C$ to $\iQ$. Section~\ref{sec:monoidal}, relying
  on the tools developed in the previous section, develops a
  methodology to devise exact conditions for lifting structures; the
  methodology is illustrated with the symmetric monoidal structure.
  The following Section~\ref{sec:closed} gives incremental conditions
  for lifting the closed structure, and Section~\ref{sec:staraut}
  gives the conditions for lifting a dualizing object, thus the \saut
  structure.
  In Section~\ref{sec:doubleNegation} we consider the case of a
  monoidal functor $Q$ into $\SLatt$, for which the results in the
  previous sections ensure that $\iQ$ is \SMC; we exhibit then a
  double negation construction, similar to the one in quantale theory,
  by which $Q$ is transformed into a functor $\Qj$ such that
  $\int \Qj$ is \saut.  We also give a representation theorem for
  those monoidal functors $Q$ into $\SLatt$ for which $\int Q$ is
  \saut via a sort of phase semantics. 
  In Section~\ref{sec:liftingfixedpoints} we study the categories of
  algebras and coalgebras (as well as their initial or \final objects)
  of an endofunctor $\lF$ of $\iQ$ that is the lifting
  of an endofunctor $F$ of $\C$.  Section~\ref{sec:examples}
  exemplifies the scope of the theory, while the last
  Section~\ref{sec:conclusions} gives concluding remarks and hints for
  future research.
  Due to the page limit, we present the proofs in the appendix.

\section{Background, and the  \Gc}
\label{sec:background}

\Pos shall denote the category of partially ordered sets (posets) and
order-preserving (or monotone) maps.  The category \SLatt has complete
lattices as objects and sup-preserving maps as arrows. It is a
subcategory of $\Pos$.  We refer the reader to standard monographs,
for example \cite{DP} and \cite{EGGHK}, for an introduction to posets
and complete lattices.
A monotone map $f:X\rto Y$ between posets has $g : Y \rto X$ as a
\emph{\ra} adjoint if and only if, for each $x\in X$ and $y \in Y$,
$f(x) \leq y$ is equivalent to $x \leq g(y)$. Such a map $f$ can have
at most one \ra for which we shall use the notation $\Star{f}$.  If
$f$ has a \ra, then $f$ is \sp and $\Star{f}$ \ip.  If $X$ and $Y$ are
complete lattices, then $f$ has a \ra if and only if it is \sp.  An
explicit formula for $\Star{f}$ is given by
$ \Star{f}(y) = \bigvee \set{x \in X \mid f(x)\leq y}$. If
$f : X \rto Y$ is a map in $\SLatt$, that is, if $f$ is \sp and $X,Y$
are complete lattices, then we can see its \ra as a \sp map
$\Star{f} : Y^{\op}\rto X^{\op}$.
A \emph{closure operator} on a complete lattice $L$ is an \opr map
$\jmath : L \rto L$ such that $x \leq \jmath(x)$ and
$\jmath(\jmath(x)) = \jmath(x)$, for all $x \in L$.
A \emph{quantale} $(Q,e,\qmult)$ is an ordered monoid whose underlying
poset is complete and whose multiplication is \sp in each
variable. Said otherwise, it is monoid in the monoidal category
$\SLatt$, or a complete posetal \SMCC, see \cite{EGGHK}.
We assume that the reader is familiar with elementary category theory
as exposed, for example, in the monograph \cite{MacLane}. We will
mainly focus on symmetric monoidal categories
$(\C,\tensor, I, \alpha,\lambda,\rho, \sigma)$, where $I$ is the unit
of the tensor $\tensor$, $\alpha$ is the associator, $\lambda$ and
$\rho$ are respectively the left and right unitors, and $\sigma$ is
the symmetry. Moreover, if $\C$ is \SMC, then we will denote by
$\funNB \impl \funNB :\C^{\op}\times \C\rto \C$ its internal
hom. Given an object $\zero$ we define a contravariant functor
$\Star{\fun}\eqdef\funNB\impl\zero$ and $\C$ is a \sautcat if the
canonical arrow $j_X:X\rto \SSX$ is invertible for every object $X$.

\medskip

In this
paper we shall
consider functors of the form $Q : \C \rto \Pos$.
\begin{definition}
  The \Gr (or total) category of $Q$, noted $\int Q$, is defined as
  follows:
  \begin{itemize}
  \item an object is a pair $(X,\alpha)$ with $X$ an object of $\C$
    and $\alpha \in Q(X)$,
  \item an arrow $(X,\alpha) \rto (Y,\beta)$ is an arrow
    $f : X \rto Y$ such that $Q(f)(\alpha) \leq \beta$.
  \end{itemize}
\end{definition}

The following lemma has an easy proof, that we skip:
\begin{lemma}
  \label{lemma:invertible}
  A morphism $f : (X,\alpha) \rto (Y,\beta)$ is invertible if and only
  if $f : X \rto Y$ is invertible and $Q(f)(\alpha) = \beta$.
\end{lemma}

The first projection $(X,\alpha) \mapsto X$  yields a functor 
$$
\pi : \int Q \rto \C\,
$$
which is the usual example of an \emph{op-fibration}. We avoid
defining here this notion which, while pervasive in many cases, tuns
out to be peripheral in our development.

\begin{commentout}
  
where the op-caretsian lifting of an arrow
$f : \pi_{1}(X,\alpha) \rto Y$ is the map $f$ itself, considered now
as a map $(X,\alpha) \rto (Y,Q(f)(\alpha))$.

\begin{definition}
  Let $p : \E \rto \B$ be a functor. An arrow $\phi : e \rto e '$ of
  $\E$ is said to be \emph{op-cartesian} if, for each arrow
  $\psi : e \rto e ''$ and $g : p(e ') \rto p(e '')$ such that
  $g \circ p(\phi) = p(\psi)$, there exists a unique
  $\chi : e ' \rto e ''$ such that $p(\chi) = g$ and $\psi = \chi
  \circ \phi$.

  A functor $p : \E \rto \B$ is called an \emph{op-fibration} if, for
  any object $e$ of $\E$ and arrow $f : p(e) \rto b$ in $\B$, there is
  an op-cartesian arrow $\phi : e \rto e '$ such that $p(\phi) = f$.

  We say that $p$ is \emph{posetal} if, for each object $b \in \B$,
  the fiber category $p^{-1}(b)$ is a poset.

\end{definition}

Let $p : \E \rto \B$ be a posetal op-fibration and let $\gamma$ be a
\emph{cleavage}, that is, the choice of an op-cartesian arrow
$\gamma_{e,f}$ for each arrow $f : p(e) \rto b$.

Then, by using the cleavage, we can define a functor
$Q : \B \rto \Pos$, sending an object $b$ to $p^{-1}(b)$ and an arrow
$f : b \rto b'$ sending $e \in p^{-1}(b)$ to the codomain $e '$ of
$\gamma_{e,f} : e \rto e '$.  Indeed, if $\psi : e \leq \tilde{e}$, then
the following diagrams are commutative in $\B$ and $\E$, respectively:
\begin{center}
  \begin{tikzcd}
    b \ar[swap]{d}{f= p(\gamma_{e,f})}\ar{r}{id}& b \ar{d}{f=p(\gamma_{\tilde{e},f})}\\
    b ' \ar{r}  \ar{r}{id} & \tilde{e}
  \end{tikzcd}
  \qquad
  \begin{tikzcd}
    e \ar{d}{\gamma_{e,f}}\ar{r}{\psi }& \tilde{e} \ar{d}{\gamma_{\tilde{e},f}}\\
    e ' \ar{r}  \ar{r}{\chi} & e ''
  \end{tikzcd}
\end{center}
where $\chi$ such that $p(\chi) = id_{b'}$ is determined by the
universal property along $\psi \circ \gamma_{e ',f} : e \rto e ''$.
Notice that this functor is an actual functor (not a pseudo functor)
since we are assuming that $p: \E \rto \B$ is posetal.
\end{commentout}

\medskip

The main aim of this paper is to describe the conditions for which some structures on the
category $\C$ (monoidal, symmetric monoidal closed, \saut, \ldots)
can be lifted to $\int Q $ in a way that the projection functor
strictly preserves the structure.

\section{Lifting functors from $\C$ to $\int Q$}
\label{sec:liftingFunctors}
Let $F : (\C^{op})^{n} \times \C^{m } \rto \C$ be a functor.
A \emph{lifting} of $F$ to $\int Q$ is a functor $\lift{F} : (\iQ^{op})^{n}
\times \iQ^{m } \rto \iQ$ such that the following diagram   strictly
commutes:
\begin{center}
  \begin{tikzcd}[ampersand replacement=\&]
    (\iQ^{op})^{n}
    \times \iQ^{m } \ar{r}{\lift{F}} \ar{d}{\pi^{n} \times \pi^{m}}\&
    \iQ \ar{d}{\pi}
    \\
     (\C^{op})^{n} \times \C^{m} \ar{r}{F} \& \C 
  \end{tikzcd}
\end{center}
A  \emph{\lent}
$\psi : \prod  Q^{n + m} \rto Q \circ F $ is a collection of 
order-preserving maps
$$\psi_{X,Y} : \prod_{i} Q(X_{i})^{op} \times \prod_{j}
Q(Y_{j}) \rto Q(F(X,Y))$$ indexed by objects $(X,Y)$ of
$\C^{n} \times \C^{m}$, such that, for each pair of maps
$f : X\rto X'$ in $\C^{n}$ and $g : Y \rto Y'$ in $\C^{m}$, the
following diagram semi-commutes:
\begin{equation}
  \label{eq:laxF}
  \adjustbox{scale=0.9,center}{
  \begin{tikzcd}[column sep=12mm,row sep=15mm]
    & \prod_{i} Q(X_{i})^{op} \times \prod_{j}Q(Y_{j}) \ar{rd}{\id \times
      \prod_{i} Q(g_{j})}\ar[swap]{ld}{\prod_{i} Q(f_{i})^{op}\times \id} \\
    \prod_{i} Q(X_{i}')^{op} \times \prod_{j } Q(Y_{j})
    \ar[swap]{d}{\psi_{X',Y}} &&\prod_{i}Q(X_{i})^{op} \times \prod_{j}Q(Y_{j}') \ar{d}{\psi_{X,Y'}}  \\
     Q(F(X', Y))   \ar{rr}{Q(F(f, g))} \ar[Rightarrow]{rru}{}
    && Q(F(X,Y')) 
  \end{tikzcd}
}
\end{equation}
By saying that the diagram above semi-commutes we mean that the lower
leg of the diagram is below the upper leg in the pointwise ordering.

Remark that the notation
$\psi : \prod  Q^{n + m} \rto Q \circ F$ is just a matter of
convenience, since for $\prod  Q^{n + m} : \C^{n + m} \rto \Pos$
while $Q \circ F : (\C^{op})^{n} \times \C^{m} \rto \Pos$.  Indeed,
$\prod$ does not make a functor from
$(C^{op})^{n} \times \C^{n} \rto \C$.
\begin{theoremEnd}
  [category=liftingFunctors]
  {theorem}
  \label{thm:lifting}
  There is a bijection beweeen liftings a functor $F$ from $\C$ to
  $\int Q$ and \lent{s}
  $\psi : \prod  Q^{n + m} \rto Q \circ F$.
\end{theoremEnd}
\begin{proofEnd}
  Let us first introduce some notation. For
  $X = (X_{1},\ldots ,X_{k})$ an object of $\C^{k}$, we write $Q(X)$
  for the product $\prod_{i = 1,\ldots ,k} Q(X_{i})$, and then, for
  $x \in Q(X) = \prod_{i} Q(X_{i})$, we write $(X,x)$ to denote the
  object $((X_{1},x_{1}),\ldots ,(X_{k},x_{k}))$ of $(\int Q)^{k}$.
  Suppose that we have a lifting $\lF$ of $F$.
  We can build $\lift{F}((X,x),(Y,y)) = (F(X,Y), z)$ with
  $z \in Q(F(X,Y))$. Calling then $z = \psi_{X,Y}(x,y)$, this defines
  a map $\psi_{X,Y} : Q(X) \times Q(Y) \rto Q(F(X, Y))$.
  We claim that $\psi_{X,Y}$ is monotone in $y \in Q(Y)$ and antitone
  in $x \in Q(X)$. For example, if $x,x' \in Q(X)$ and $x' \leq x$,
  then $\id_{X}$ is a map in $(\iQ)^{n}$ from $(X,x')$ to $(X,x)$; then,
  by functoriality and since $\lF$ is contravariant in $X$,
  $\lF(\id_{X},(Y,y)) $ is a map in $\iQ$ from
  $\lF((X,x),(Y,y))$ to
  $\lF((X,x'),(Y,y))$,
  which amounts to the inequality
  $$
  \psi_{X,Y}(x,y) = Q(F(\id_{X}, \id_{Y}))(\psi_{X,Y}(x,y)) \leq
  \psi_{X,Y}(x',y)\,.$$
  Next, if $f = (f_{1},\ldots,f_{n}): (X,x) \rto (X',x')$ and
  $g = (g_{1},\ldots,g_{m}) : (Y,y) \rto (Y',y')$ are maps in
  $(\iQ)^{n}$ and $(\iQ)^{m}$, then
  $\lF(f,g) : \lF((X',x'),(Y,y)) \rto \lF((X,x),(Y',y'))$ is a map of
  $\iQ$.  This statement, an entailment, is equivalent to the
  statement $Q(f)(x) \leq x'$ and $Q(g)(y) \leq y'$ implies
  $Q(F(f,g))(\psi_{X',Y}(x',y)) \leq \psi_{X,Y'}(x,y')$ for all
  $x,x',y,y'$. Letting $x' = Q(f)(x) $ and $y' = Q(g)(y)$ in the last
  inequality, we deduce the inclusion
  \begin{align*}
    Q(F(f, g))(\psi_{X',Y}(Q(f)(x),y)) & \leq
    \psi_{X,Y'}(x,Q(g)(y))\,,
  \end{align*}
  which amounts to the semi-commutativity of \eqref{eq:laxF}.  On the
  other hand, this inclusion is sufficient as well: if
  $Q(f)(x) \leq x'$ and $Q(g)(y) \leq y'$, then, by the order
  properties of $\psi$, we have
  \begin{align*}
    Q(F(f,g))(\psi_{X',Y}(x',y)) & \leq
    Q(F(f,g))(\psi_{X,Y}(Q(f)(x),y)) \\
    & \leq \psi_{X,Y'}(x,Q(g)(y))\leq \psi_{X,Y'}(x,y')\,.
  \end{align*}
  Consequently, the contruction of $\psi$ from $\lF$ can be reversed,
  thus building a lifting of $F$ to $\int Q$ from a collection $\psi$
  with the stated properties.
\end{proofEnd}

We shall say that \emph{$Q$ factors through
  $\SLatt$} if
\begin{enumerate}[(a)]
\item each $Q(X)$ is complete lattice,
\item for $f : X \rto Y$, $Q(f) : Q(X) \rto Q(Y)$ is \sp.
\end{enumerate}
The more advanced reader will have recognised that $Q$ factors through
$\SLatt$ means indeed that we can write $Q = U \circ Q_{0}$ for a
functor $Q_{0} : \C \rto \SLatt$, where $U : \SLatt \rto \Pos$ is
the inclusion functor.  If $Q$ factors through $\SLatt$, then
(lax) extranaturality can be replaced by (lax) naturality, as stated
in the next Lemma.

\begin{theoremEnd}
  [category=liftingFunctors]
  {lemma}
  Suppose $Q$ factors through \SLatt. Then 
   any \lent{s}
  $\psi : \prod  Q^{n + m} \rto Q \circ F$
  makes the diagram below 
  semi-commutative.
  \begin{equation}
    \nonumber
    \begin{tikzcd}[column sep=15mm,row sep=12mm,ampersand replacement= \&]
      \prod_{i} Q(X_{i}')^{op} \times \prod_{j } Q(Y_{j})
      \ar[swap]{d}{\psi_{X',Y}} 
      \ar{rr}{\prod_{i} \Star{Q(f_{i})} \times
        \prod_{i} Q(g_{j})}
      \&\&\prod_{i}Q(X_{i})^{op} \times \prod_{j}Q(Y_{j}') 
      \ar{d}{\psi_{X,Y'}}  \\
       Q(F(X', Y))   \ar{rr}{Q(F(f, g))} 
       \ar[Rightarrow]{rru}{}
      \&\& Q(F(X,Y')) 
    \end{tikzcd}
  \end{equation}
\end{theoremEnd}

\begin{proofEnd}
  Just like before, the maps $\psi_{X,Y}$ are antitone on the first
  variable and monotone on the second. If $f :(X,x) \rto (X',x')$ is
  in $\int Q$, then $Q(f)(x) \leq x'$ which as $Q(f)$ is a
  sup-preserving map is equivalent to $x \leq \Star{Q(f)}(x')$. If
  moreover $g : (Y,y) \rto (Y',y')$ -- ie $Q(g)(y) \leq y'$ -- then
  $Q(F(f,g))(\psi_{X',Y}(x',y)) \leq \psi_{X,Y'}(x,y')$ for all
  $x,x',y,y'$. Again, by instancing this last statement on
  $x =\Star{Q(f)}(x')$ and $y' =Q(g)(y)$ and by the order properties
  of $\psi_{X,Y}$, we have that
\begin{align*}
Q(F(f,g))(\psi_{X',Y}(x',y)) \leq
\psi_{X,Y'}(\Star{Q(f)}(x'),Q(g)(y))\,.
\tag*{\qedhere}
\end{align*}
\end{proofEnd}
\noindent
Note that the $\psi_{X,Y}$ need not be \sp, even when $Q$ factors
through $\SLatt$.

\section{Lifting the monoidal structure}
\label{sec:monoidal}

Theorem \ref{thm:lifting} immediately tells how to lift the tensor $\tensor$, as a
simple bifunctor, from $\C$ to $\int Q$.

\begin{proposition}
  There is a bijection between the following kind of data:
  \begin{itemize}
  \item a lifting of a bifunctor $\otimes $ from $\C$ to $\int Q$,
    
  \item a collection of order-preserving maps
    $$\set{\mu_{X,Y} : Q(X) \times Q(Y) \rto Q(X \tensor Y) \mid X,Y \in
      \Obj(\C)}\,,$$ such that, for $f : X \rto X'$ and $g : Y \rto Y'$,
      the following diagram semi-commutes:
    \begin{equation}
      \label{eq:laxmonoidal}
      \begin{tikzcd}[column sep=20mm]
        Q(X) \times Q(Y) \ar{r}{Q(f) \times
          Q(g)}\ar[swap]{d}{\mu_{X,Y}} & Q(X') \times Q(Y') \ar{d}{\mu_{X',Y'}}\\
        Q(X \tensor Y) \ar{r}{Q(f \tensor g)}  \ar[Rightarrow]{ru}{} 
         & Q(X' \tensor Y')
      \end{tikzcd}
    \end{equation}
  \end{itemize}
\end{proposition}
Let us remark that the maps $\mu_{X,Y}$ are natural (i.e. diagram
\eqref{eq:laxmonoidal} fully commutes) if and only if the lifted
bifunctor preserves op-cartesian arrows, that is the projection
functor $\iQ \rto \C$ is a monoidal op-fibration as defined in
\cite{MV2020}.  In a similar way, lifting the unit (of a tensor
structure) as a constant functor yields $u \in Q(I)$ or, equivalently,
$u : 1 \rto Q(I)$.

Next, lifting the maps ($\lambda, \rho, \alpha, \sigma$) of a
symmetric monoidal structure on $\C$ simply amounts to requiring the
following inclusions:
\begin{align*}
  Q(\lambda)(\mu_{I,Y}(u,y))& \leq y \,, &
  Q(\rho)(\mu_{X,I}(x,u)) & \leq x\,,\\
  Q(\alpha)(\mu_{X\otimes Y,Z}(\mu_{X,Y}(x,y),z)) & \leq \mu_{X,Y
    \otimes Z}(x,\mu_{Y,Z}(y,z)) \,, \\ 
  Q(\sigma)(\mu_{X,Y}(x,y)) & \leq \mu_{Y,X}(y,x)\,.
\end{align*}

Considering that we actually need these lifted maps to be invertible,
we require (by Lemma~\ref{lemma:invertible}) the above inclusions to
be equalities. Thus we have $u :1 \rto Q(I)$ and
$\mu_{X,Y} : Q(X) \times Q(Y) \rto Q(X \tensor Y)$, where the latter
is lax natural,
making commutative the usual diagrams for a monoidal functor, see
Figure~\ref{fig:monoidalFunctor}.
\begin{figure}[h]
  \begin{equation}
    \nonumber
    \begin{tikzcd}
    1 \times Q(X) \ar{d}{\lambda_{Q}}  \ar{r}{u \times id}&
    Q(I) \times Q(X) \ar{d}{\mu} \\
    Q(X) & Q(I \tensor X) \ar{l}{Q(\lambda)}
  \end{tikzcd}
  \qquad
  \begin{tikzcd}
    Q(X) \times 1 \ar{d}{\rho_{Q}}  \ar{r}{id \times u}&
    Q(X) \times Q(I) \ar{d}{\mu} \\
    Q(X) & Q(X \tensor I) \ar{l}{Q(\rho)}
  \end{tikzcd}
  \end{equation}
  \begin{equation}
    \nonumber    
    \begin{tikzcd}[column sep=0.8mm]
      (Q(X) \times Q(Y)) \times Q(Z) \ar{r}{\alpha_{Q}}
      \ar{d}{\mu \times id}
      &  Q(X) \times (Q(Y) \times
      Q(Z)) \ar{d}{id \times \mu}\\
      Q(X \tensor Y) \times Q(Z)     \ar{d}{\mu}
      &  Q(X) \times Q(Y \tensor
      Z) \ar{d}{\mu} \\
      Q((X \tensor Y) \tensor Z) \ar{r}{Q(\alpha)} &  Q(X \tensor (Y \tensor
      Z)) \\
    \end{tikzcd}
    \hspace{-8mm}
    \begin{tikzcd}[row sep=12mm]
      Q(X) \times Q(Y) \ar{r}{\sigma_{Q}}  \ar{d}{\mu_{X,Y}}&
      Q(Y) \times Q(X) \ar{d}{\mu_{Y,X}} \\
      Q(X \tensor Y) \ar{r}{Q(\sigma)} & Q(Y \tensor X) 
    \end{tikzcd}    
  \end{equation}
  \caption{Coherence diagrams for a monoidal functor.}
  \label{fig:monoidalFunctor}
\end{figure}

\section{Lifting the closed structure}
\label{sec:closed}

We suppose next that we have lifted the symmetric monoidal structure
from $\C$ to $\int Q$ via a lax natural $\mu$ as in the previous
section, and that $\C$ is closed as well. Let us recall that $\C$ is
closed if, for each object $X$ of $\C$, the functor $X \tensor -$
(tensoring with $X$) has a right adjoint, noted here $X \impl -$. As
from elementary theory, the right adjoint is made into a bifunctor
$\impl : \C^{\op} \times \C \rto \C$. Moreover, for each pair of
objects $X$ and $Y$ of $\C$, we have maps (the units and counits of
the adjunction)
\begin{align*}
  \eta_{X,Y}  &  : Y \rto X \impl (X \tensor Y)\,, & \ev_{X,Y} & : X \tensor (X \impl Y) \rto Y\,, 
\end{align*}
natural in $Y$, satisfying the usual pasting diagrams for adjunctions:
\begin{align}
  (X \impl \ev_{X,Y}) \circ \eta_{X,X \impl Y} & = \id_{X \impl Y}\,,
  &
  \ev_{X,X \tesnor Y} \circ (X \tensor \eta_{X,Y}) & = \id_{X \tensor Y}\,.\label{eq:adj}
\end{align}
The following statement, allowing to lift the bifunctor $\impl$, is a
direct consequence of Theorem~\ref{thm:lifting}.
\begin{proposition}
  There is a bijection between the following kind of data:
  \begin{itemize}
  \item a lifting of a bifunctor $\impl: \C^{op}\times \C \rto \C$ to a
    functor $(\int Q)^{op} \times \int Q \rto \int Q$,
    
  \item a collection of order-preserving maps
    $$\set{\iota_{X,Y} : Q(X)^{op} \times Q(Y) \rto Q(X \impl Y) \mid X,Y \in
      \Obj(\C)}\,,$$ such that, for $f : X \rto X'$ and $g : Y \rto Y'$,
      the following diagram semi-commutes:
    \begin{equation}
      \label{eq:laxImpl}
      \begin{tikzcd}[column sep=15mm]
        Q(X)^{op} \times Q(Y) \ar{r}{\id \times
          Q(g)}\ar[swap]{d}{Q(f)^{op}\times \id}
        &Q(X)^{op} \times Q(Y') \ar{dd}{\iota_{X,Y'}}  \\
        Q(X')^{op} \times Q(Y) \ar[swap]{d}{\iota_{X',Y}} & \\
        Q(X' \impl Y)   \ar{r}{Q(f \impl g)} \ar[Rightarrow]{ruu}{}
        & Q(X \impl Y') 
      \end{tikzcd}
    \end{equation}
  \end{itemize}
\end{proposition}

\begin{commentout}
  \begin{remark}
    For $f = \id_{X} : X \rto X$ in \eqref{eq:laxImpl} we obtain
    semi-commutativity of the following diagram:
    \begin{equation}
      \label{eq:laxImpl2}
      \begin{tikzcd}[column sep=15mm]
        Q(X)^{op} \times Q(Y) \ar[swap]{d}{\iota_{X,Y}} \ar{r}{\id
          \times Q(g)} 
        &Q(X)^{op} \times Q(Y') \ar{d}{\iota_{X,Y'}}  \\
        Q(X \impl Y) \ar{r}{Q(X \impl g)} \ar[Rightarrow]{ru}{} & Q(X
        \impl Y')
      \end{tikzcd}
    \end{equation}
  \end{remark}
\end{commentout}

Once the bifunctor $\impl$ has been lifted, in order to lift the
adjunction $X \tensor \funNB \dashv X \impl \funNB$, we require the unit and
counit to be maps of $\iQ$. This is achieved by requiring the
following inequalities to hold, for each pair of objects $X,Y$ of
$\C$, each $x \in Q(X)$ and $y \in Q(Y)$:
\begin{align}
  \label{eq:unit0}
  Q(\eta_{X,Y})(y) & \leq \iota_{X,X \tensor Y}(x,\mu_{X,Y}(x,y))\,,
  \\
  \label{eq:counit0}
  Q(\ev_{X,Y})(\mu_{X,X \impl Y}(x,\iota_{X,Y}(x,y)) & \leq y\,.
\end{align}

Due to its central role in this paper, let us introduce an explicit
notation for  the map
\begin{align*}
  \brack{\,\funNB, \,\funNB\,}_{X,Y} & \eqdef  Q(\ev_{X,Y}) \circ \mu_{X,X
    \impl Y} :  Q(X) \times Q(X \impl Y) \rto Q(Y)\,.
\end{align*}
Notice  that the inclusion~\eqref{eq:counit0} is then written as
$\brack{x,\iota_{X,Y}(x,y)}_{X,Y} \leq y$.

\begin{theoremEnd}[category=closed]{theorem}
  \label{thm:main}
  The following conditions are equivalent:
  \begin{enumerate}[(i)]
  \item $\C$ and $\iQ$ are symmetric monoidal closed and
    $\pi : \iQ \rto \C$ strictly preserves this structure. \label{thm:enum:C and iQ smcc}
  \item   
  For each pair of objects $X,Y$ of $\C$, and each $x \in Q(X)$ and $y
  \in Q(Y)$, the equality
  \begin{align}
    \label{eq:mudefinable}
    \mu_{X,Y}(x,y) & = \brack{\,x,\,Q(\eta_{X,Y})(y)\,}_{X,X \tensor Y}
  \end{align}
  holds
  and, for each $\alpha \in Q(X)$, the map
  \begin{align}
    \brack{\alpha,\,\funNB\,}_{X,Y} & : Q(X \impl Y) 
    \rto Q(Y)
    \label{map:needRightAdj}
  \end{align}
  has a right adjoint $\iota_{X,Y}(\alpha,\,\funNB\,)$.\label{thm:enum:right adjoint de eval}
  \end{enumerate}
\end{theoremEnd}

\begin{proofEnd}
  Firstly, let us recall that a map
  $ \iota_{Y}^{\alpha} : Q(Y) \rto Q(X \impl Y) $ is the right adjoint of  $\brack{\alpha,\,\cdot\,}_{X,Y}$  if and only if the usual
  unit-counit relations
  \begin{align}
    \label{eq:unitcounit1}
    \beta & \leq\iota^{\alpha}_{Y}(\,\brack{\alpha,\beta}_{X,Y}\,)\,,
    &
    \brack{\,\alpha,\,\iota^{\alpha}_{Y}(\gamma)\,}_{X,Y} & \leq \gamma \,,
  \end{align}
  hold, for each $\beta \in Q(X \impl Y)$ and $\gamma \in Q(Y)$.
  \newline
  Assume (\ref{thm:enum:C and iQ smcc}).  Let
  $\iota_{X,Y} : Q(X)^{op}\times Q(Y) \rto Q(X \impl Y)$ satisfying
  \eqref{eq:unit0} and \eqref{eq:counit0} be given. Let us show first
  that \eqref{eq:mudefinable} holds.  Recall that
  $Q(X \tensor \eta_{X,Y})(\mu_{X,Y}(x,y)) \leq \mu_{X,X \impl (X
    \tensor Y)}(x,Q(\eta_{X,Y})(y))$. Then
  \begin{align*}
    \mu_{X,Y}(x,y) & = Q(\ev_{X,X \tensor Y})(\,Q(X \tensor
    \eta_{X,Y})(\mu_{X,Y}(x,y))\,) \\
    & 
    \leq Q(\ev_{X,X \tensor Y})(\,\mu_{X,X \impl
      (X \tensor Y)}(x,Q(\eta_{X,Y})(y))\,) =
    \brack{x,Q(\eta_{X,Y})(y)}_{X,X \tensor Y}\,.
  \end{align*}
  For the opposite inclusion, observe that, using \eqref{eq:unit0} and \eqref{eq:counit0},
  \begin{align*}
    \brack{\,x,\,Q(\eta_{X,Y})(y)\,}_{X,X \tensor Y} & \leq
    \brack{\,x,\,\iota_{X,X \tensor Y}(x,\mu_{X,Y}(x,y))\,}_{X,X
      \tensor Y}
    \leq \mu_{X,Y}(x,y)\,. 
  \end{align*}
  Next, let $\alpha \in Q(X)$ and
  define $\iota_{Y}^{\alpha}(\beta)$ as $\iota_{X,Y}(\alpha,\beta)$.
  Then \eqref{eq:counit0} immediately yields the counit of
  \eqref{eq:unitcounit1}.  For the unit, we
  argue 
  as follows:
  \begin{align*}
    \beta & = (Q(X \impl \ev_{X,Y}) \circ Q(\eta_{X,X \impl
      Y}))(\beta) \\
    & \leq Q(X \impl \ev_{X,Y})( \iota_{X,X \tensor (X \impl
      Y)}(\alpha,\mu_{X,X \impl Y}(\alpha,\beta)))\,, \tag*{using
      \eqref{eq:unit0},}
    \\
    & \leq \iota_{X,Y}(\alpha,Q(\ev_{X,Y})(\mu_{X,X \impl
      Y}(\alpha,\beta))) \,, \tag*{using
      \eqref{eq:laxImpl},} \\
    & = \iota_{Y}^{\alpha}(\brack{\alpha,\beta}_{X,Y})\,.
  \end{align*}
  Assume now (\ref{thm:enum:right adjoint de eval}), so $\iota_{Y}^{\alpha}$ is given for each
  $\alpha \in Q(X)$. We define $\iota_{X,Y}(\alpha,\beta)$ as
  $\iota_{Y}^{\alpha}(\beta)$ and verify that \eqref{eq:counit0} and
  \eqref{eq:unit0} hold.  Again, the co-unit in \eqref{eq:unitcounit1} is
  exactly equation \eqref{eq:counit0}. Thus, we only need to derive
  \eqref{eq:unit0}, 
  which we do as follows:
  \begin{align*}
    Q(\eta_{X,Y})(y) & \leq \iota_{X \tensor
      Y}^{\alpha}(\,\brack{\alpha,Q(\eta_{X,Y})(y)}_{X, X \tesnor Y}\,)\,,
    \tag*{letting $\beta =
      Q(\eta_{X,Y})(y) \in Q(X \impl (X \tensor Y))$ in the unit of \eqref{eq:unitcounit1},}
    \\
    & = \iota_{X \tensor
      Y}^{\alpha}(\,\mu_{X,Y}(\alpha,y)\,)
    =  \iota_{X,X \tensor
      Y}(\alpha,\mu_{X,Y}(\alpha,y))\,,
    \tag*{using~\eqref{eq:mudefinable}. \qedhere}
  \end{align*}
\end{proofEnd}

\begin{commentout}
  \begin{remark}
    A consequence of commutativity of diagram~\eqref{diag:mueta} is
    the following.  For $f : X \tensor Y \rto Z$, suppose that
    $Q(X \tensor X \impl f) \circ \mu_{X,X \impl Y} = \mu_{X,Z} \circ
    (Q(X) \times Q(X \impl f))$.  

    Then, for $f^{\sharp} : Y \rto X \impl Z$ the exponential
    transpose of $f : X \tensor Y \rto Z$, we have
    $Q(X \tensor f^{\sharp}) \circ \mu_{X,Z} = \mu_{X,Z} \circ (Q(X)
    \times Q(f^{\sharp}))$.  Indeed, we can write
    $f^{\sharp} = (X \impl f) \circ \eta_{X,Y}$, and then
    \begin{center}
      \begin{tikzcd}[column sep=20mm]
        Q(X) \times Q(Y) \ar{r}{Q(X) \times Q(\eta_{X,Y})}
        \ar{d}{\mu_{X,Y}} & Q(X) \times Q(X \impl (X \tensor Y))
        \ar{d}{\mu_{X,X \impl (X \tensor Y)}} \ar{r}{Q(X) \times Q(X
          \impl f)}
        & Q(X) \times Q(X \impl Z) \ar{d}{\mu_{X,X \impl Z}} \\
        Q(X \tensor Y) \ar{r}{Q(X \tensor \eta_{X,Y})} & Q(X \tensor
        (Y \impl (X \tensor Y))) \ar{r}{Q(X \tensor X \impl f)} & Q(X
        \tensor X \impl Z)
        \\
      \end{tikzcd}
    \end{center}
  \end{remark}
\end{commentout}
We notice that if $\mu$ is natural (and not merely lax natural) in its
right variable, then equation \eqref{eq:mudefinable} automatically
holds, as evident from the proof of Theorem~\ref{thm:main}. 

\medskip

We say that \emph{$Q$ monoidally factors through $\SLatt$} if $Q$
factors through $\SLatt$ and, moreover,
$\mu_{X,Y} : Q(X) \times Q(Y) \rto Q(X \tensor Y)$ is natural (and not
merely lax natural) and bilinear, that is, \sp in each variable,
separately.  That is, writing $Q= U \circ Q_0$, the extension of
$\mu_{X,Y}$ to the tensor product $Q_0(X) \tensor Q_0(Y)$ in $\SLatt$
makes $Q_0 :\C \rto \SLatt$ into a monoidal functor.
Theorem~\ref{thm:main}
immediately implies the following statement.
\begin{theorem}
  \label{thm:SLattClosed}
  If $Q$ monoidally factors through $\SLatt$, then $\iQ$ is \SMC and
  the projection functor $\pi : \iQ \rto \C$ strictly preserves this
  structure.
\end{theorem}

It can be expected that if $\mu$ is natural, then $\iota$ has also
some kind of naturality. This is the content of the next lemma and
corollary.
\begin{theoremEnd}
  [category=closed]
  {lemma}
  \label{lemma:iotaNatural}
  Suppose that $\mu$ is natural.  Let $f : X \rto Y$ and suppose that
  $Q(f \impl Z)$ has a right adjoint $Q(f \impl Z)^{\ast}$.  Then, for
  $\alpha \in Q(X)$ and $\gamma \in Q(Z)$, the following holds:
  \begin{align}
    \iota_{Y,Z}(Q(f)(\alpha),\gamma) & = Q(f \impl
    X)^{\ast}(\iota_{X,Z}(\alpha, \gamma ))\,.
    \label{eq:iotaRA}
  \end{align}
\end{theoremEnd}
\begin{proofEnd}
  Observe first that the following diagram commutes:
  \begin{center}
    \adjustbox{scale=0.9}{
    \begin{tikzcd}[column sep=20mm, ampersand replacement=\&]
      Q(X) \times Q(Y \impl Z) \ar[swap]{d}{Q(f) \times Q(Y \impl Z)} \ar{r}{Q(X) \times Q(f \impl Z)}
      \ar{rd}{\mu_{X,Y \impl Z}} \& Q(X) \times Q(X \impl Z) \ar{rd}{\mu_{X,X \impl Z}} \\
      Q(Y) \times Q(Y \impl Z) \ar[swap]{rd}{\mu_{Y,Y \impl Z}}
      \& Q(X \tensor (Y \impl Z)) \ar{d}{Q(f \tensor (Y \impl Z))} \ar{r}{Q(X \tensor (f \impl Z))} \& Q(X \tensor (X \impl Z)) \ar{d}{Q(\ev_{X,Z})} \\
      \& Q(Y \tensor (Y \impl Z)) \ar{r}{Q(\ev_{Y,Z})} \& Q(Z)
    \end{tikzcd}
  }
  \end{center}
  With this in mind, let us compute as follows (with $\beta \in Q(Y
  \impl Z)$):
  \begin{align*}
    \beta \leq \iota_{Y,Z}(Q(f)(\alpha),\gamma) & \Tiff
    Q(\ev_{Y,Z})(\mu(Q(f)(\alpha),\beta)) \leq \gamma \\
    & \Tiff
    Q(\ev_{X,Z})(\mu(\alpha,Q(f \impl Z)(\beta))) \leq \gamma \\
    & \Tiff
    Q(f \impl Z)(\beta) \leq \iota_{X,Z}(\alpha,\gamma) \\
    & \Tiff
    \beta \leq Q(f \impl Z)^{\ast}(\iota_{X,Z}(\alpha,\gamma))
    \tag*{\qedhere}\,.
  \end{align*}
\end{proofEnd}
Notice that Lemma~\ref{lemma:iotaNatural} applies when $f : X \rto Y$
is invertible, in which case
$Q(f \impl X)^{\ast} = Q(f^{-1} \impl X)$.

\begin{commentout}
  \begin{corollary}
    Suppose that $\mu$ is natural, $\C$ is a monoidal closed posetal
    $2$-category, $Q$ is a $2$-functor, and $f : X \rto Y$ has
    $f^{\ast} : Y \rto X $ as right adjoint.
    \begin{align*}
      \iota_{X,Z}(Q(f^{\ast})(\alpha),\gamma) & = Q(f \impl
      Z)(\iota_{Y,Z}(\alpha,\gamma))\,,
    \end{align*}
    for each $\alpha \in Q(X)$ and $\gamma \in Q(Z)$.
  \end{corollary}
  As from Lemma~\ref{lemma:iotaNatural}, considering that
  $f^{\ast} \impl Z : X \impl Z \rto Y \impl Z$ is left adjoint to
  $f \impl Z : Y \impl Z \rto X \impl Z$, $Q(f^{\ast} \impl Z)$ is
  left adjoint to $Q(f \impl Z)$, since $Q$ preserves adjunctions.
\end{commentout}

\medskip

In the following statement, we fix $0 \in \C$ and $\omega \in Q(0)$,
let $\SX \eqdef X \impl 0$ and $\omega_{X} : Q(X) \rto Q(\SX)^{\op}$ be
defined by $\omega_{X}(\alpha) \eqdef \iota_{X,0}(\alpha,\omega)$. We use
$\fun^{\op} : \SLatt^{op} \rto \SLatt$ to denote the functor sending a
complete lattice to its opposite and a \sp map to its right adjoint.
Equation~\eqref{eq:iotaRA} implies commutativity of the diagram
\begin{center}
  \begin{tikzcd}[ampersand replacement=\&, column sep=18mm]
    Q(X) \ar[swap]{d}{\omega_{X} = \iota_{X,0}(\funNB,\omega)} \ar{r}{Q(f)} \& Q(Y) \ar{d}{\omega_{Y} = \iota_{Y,0}(\funNB,\omega)} \\
    Q(\SX)^{\op} \ar{r}{Q(\Star{f})^{\op}} \& Q(\SY)^{\op}\,.
  \end{tikzcd}
\end{center}
\begin{corollary}
  \label{cor:omegaNatural}
  Suppose that $Q$ monoidally factors through $\SLatt$. 
  Then, considering $Q$ as a functor $\C \rto \SLatt$, $\omega_X$ is a
  natural transformation from $Q(X)$ to $Q(\SX)^{\op}$.
\end{corollary}

\begin{commentout}
  \begin{theoremEnd}[category=closed]{lemma}
    Suppose that $\mu$ is natural, $\C$ is a monoidal closed posetal
    $2$-category, $Q$ is a $2$-functor, and $f : Y \rto Z$ is a right
    adjoint. Then
    \begin{align*}
      \iota_{X,Z}(\alpha,Q(f)(\beta)) & = Q(X \impl
      f)(\iota_{X,Y}(\alpha,\beta))\,,
    \end{align*}
    for each $\alpha \in Q(X)$ and $\beta \in Q(Y)$.
  \end{theoremEnd}
  \begin{proofEnd}
    Let $f_{\ast} : Z \rto Y$ be such that $f_{\ast} \dashv f$.  Since
    $Q$ is a 2-functor, $Q(f_{\ast}) \dashv Q(f)$ and
    $Q(X \impl f_{\ast}) \dashv Q(X \impl f)$.Then
    $Q(X \impl f)(\iota_{X,Y}(\alpha,\cdot))$ is right adjoint to
    \begin{align*}
      Q(\ev_{X,Y})(\mu_{X,X \impl Y}(\alpha,Q(X \impl
      f_{\ast})(\cdot))) & = Q(\ev_{X,Y} \circ (X \tensor (X \impl
      f_{\ast})))(\mu_{X,X \impl
        Z}(\alpha,\cdot)) \\
      & = Q(f_{\ast} \circ \ev_{X,Z})(\mu_{X,X \impl
        Z}(\alpha,\cdot)) \\
      & = Q(f_{\ast}) \circ Q(\ev_{X,Z})(\mu_{X,X \impl
        Z}(\alpha,\cdot))
    \end{align*}
    where the last is left adjoint to
    $\iota_{X,Z}(\alpha,Q(f)(\cdot))$.
  \end{proofEnd}
\end{commentout}

\begin{commentout}
  \subsection{Modules}
  Suppose that $\mu$ is natural, so $Q$ is monoidal. Then $Q$ lifts
  the category of monoids in $\C$ to the category of monoids in
  $\Pos$, and restricts to the respective categories of commutative
  monoids. It follows that $Q(I)$ is a commutative ordered monoid,
  whose structure is given by
  \begin{align*}
    & Q(I) \times Q(I) \rto[\mu_{I,I}] Q(I \tensor I) \rto[Q(\rho_{I})
    = Q(\lambda_{I})] Q(I)\,, \intertext{and that $Q(X)$ is a (left or
      right) $Q(I)$-module with action} & Q(I) \times Q(X)
    \rto[\mu_{I,X}] Q(I \tensor X) \rto[Q(\lambda_{X})] Q(X)\,.
  \end{align*}
  It is then immediate to observe that every arrow $f : X \rto Y$
  leads to an equivariant map $Q(f) : Q(X) \rto Q(Y)$.
\end{commentout}

\section{Lifting dualizing objects}
\label{sec:staraut}
We incrementally assume that $\int Q$ and $\C$ are \SMC and that the
projection strictly preserves all the structure.  We
use 
the notation introduced before. Namely, for a fixed object $0$ of a
\SMCC, we let $\SX \eqdef X \impl 0$. If $0$ and $X$ are objects of
$\C$ and $\omega \in Q(0)$, $\omega_{X} : Q(X)^{op} \rto Q(\SX)$ is
defined by $\omega_{X}(\alpha) \eqdef \iota_{X,0}(\alpha,\omega)$.

\smallskip

Let us recall that an object $0$ of a \SMCC $\V$ is \emph{dualizing}
if, for each object $X$ of $\V$, the canonical arrow
\begin{align*}
  j_{X} & : X \rto \SSX \,,
\end{align*}
arising as the transpose the map $\ev_{X,0} \circ \sigma_{\SX,X} : \SX
\tensor X \rto 0$,
is an isomorphism.
A \SMCC $\V$ is \emph{\saut} if it comes with a dualizing object $0$.

\smallskip

We devise in this section conditions for $(0,\omega)$ being a
dualizing object of $\iQ$. Recall from Section \ref{sec:closed} that
$\Star{(X,\alpha)}\eqdef (X,\alpha) \impl (0,\omega) =
(\SX,\omega_X(\alpha))$, and notice that since $\iQ$ is closed the
canonical arrow $j_{(X,\alpha)}$ 
is part of this structure and projects to $j_{X}$. As this might not
be evident, we explicitly record this fact below. 
\begin{theoremEnd}[category=staraut]{lemma}
  For each $X$ object $X$ of $\C$ and each $\alpha \in Q(X)$, the
  canonical arrow $j_X : X \rto \SSX$ lifts to an arrow
  $j_{(X,\alpha)} : (X,\alpha) \rto
  (X,\alpha)^{**}=(X^{**},\omega_{\SX}(\omega_X(\alpha)))$ in
  $\int Q$. That is, we have the inequality 
  $Q(j_X)(\alpha) \leq \omega_{\SX}(\omega_X(\alpha))$.
\end{theoremEnd}
\begin{proofEnd}
  Writing
  $j_X=(\SX \impl (\ev_{X,0} \circ \sigma)) \circ \eta_{\SX,X}$, that
  is general formula for transpose of an adjunction, we have for
  $\alpha \in Q(X)$
  \begin{align*}
    Q(j_X)(\alpha) & = Q((\SX \impl (\ev_{X,0} \circ \sigma)) \circ
    \eta_{\SX,X})(\alpha)
    \\
    & \leq Q(\SX \impl (\ev_{X,0} \circ \sigma))(\iota_{\SX,\SX\tensor
      X}(x,\mu_{\SX,X}(x,\alpha))) \tag*{ for all
      $x$, by functoriality and the inclusion~\eqref{eq:unit0}}
    \\
    & \leq \iota_{\SX,0}(x,Q(\ev_{X,0} \circ
    \sigma)(\mu_{\SX,X}(x,\alpha))) \tag*{by
      diagram~\eqref{eq:laxImpl}}
    \\
    & = \iota_{\SX,0}(x,Q(\ev_{X,0})(\mu_{X,\SX}(\alpha,x))) \tag*{by
      functoriality and figure~\eqref{fig:monoidalFunctor}}
    \\
    & \leq \iota_{\SX,0}(\iota_{X,0}(\alpha,\omega),\omega) \tag*{with
      $x=\iota_{X,0}(\alpha,\omega)$ and equation~\eqref{eq:counit0}}
    \\
    &= \omega_{\SX}(\omega_X(\alpha)) \,.
    \tag*{\qedhere}
  \end{align*}
\end{proofEnd}
We study next under which conditions the maps $j_{(X,\alpha)}$ are
isomorphisms.
Using Lemma~\ref{lemma:invertible} and assuming that
$Q(X) \neq \emptyset$, for each object $X$ of $\C$, we infer the
following:

\begin{theoremEnd}[category=staraut]{proposition}
  \label{prop:dualizing}
  An object $(0,\omega)$ of $\iQ$ is dualizing if and only if $0$ is
  a dualizing object of $\C$ and
  , for each object $X$ of $\C$, the
  following diagram commutes.
  \begin{equation}
    \label{diag:staraut}
    \begin{tikzcd}[column sep=20mm,ampersand replacement=\&]
      Q(X) \ar[swap]{rd}{Q(j_{X})}\ar{r}{\omega_{X}} \& Q(\SX)^{op}
      \ar{d}{\omega_{\SX}} \\
      \& Q(\SSX)
    \end{tikzcd}
  \end{equation}
\end{theoremEnd}
\begin{proofEnd}
  Asking for the object $(0,\omega)$ in $\int Q$ to be dualizing means
  that for each object $X$ of $\C$ and $\alpha \in Q(X)$, the arrow
  $j_{(X,\alpha)} : (X,\alpha) \rto (X,\alpha)\Star{}\Star{}$ is
  invertible. By Lemma~\ref{lemma:invertible} and since we can always
  find $\alpha \in Q(X)$, this is equivalent to asking $j_X$ to be
  invertible, for each object $X$ of $\C$, and that
  $Q(j_X)(\alpha) = \omega_{\SX}(\omega_{X}(\alpha))$, for each
  object $X$ of $\C$ and $\alpha \in Q(X)$.
\end{proofEnd}

\begin{theoremEnd}[category=staraut]{lemma}
  \label{lemma:easyDirect}
  If $(0,\omega)$ is dualizing, then, for each object $X$ of $\C$,
  $\omega_{X} : Q(X)^{op} \rto Q(\SX)$ is an isomorphism.
\end{theoremEnd}
\begin{proofEnd}
  In diagram~\eqref{diag:staraut}, since $Q(j_{X})$ is invertible,
  $\omega_{X}$ is split mono, and $\omega_{\SX}$ is split epi, for
  each object $X$. Thus, $\omega_{\SX}$ is both split epi and split
  mono and therefore it is invertible. It follows that $\omega_{X}$ is
  inverted by $Q(j_{X})^{-1} \circ \omega_{\SX} $.
\end{proofEnd}

It might be difficult to directly make use of
Proposition~\ref{prop:dualizing}. The following
Theorem~\ref{thm:omegaInv} yields a more direct way to verify whether
some $(0,\omega)$ is dualizing by simply looking at the maps
$\omega_X$.  Recall from Theorem~\ref{thm:main} that any map of the
form $\brack{\alpha,\funNB}_{X,Y}$ has a right adjoint, namely
$\iota_{X,Y}(\alpha,\funNB)$. In the case $Y = 0$, let us simply write
$\brack{\alpha,\funNB}_{X}$ for $\brack{\alpha,\funNB}_{X,0}$, so to 
have the equivalence:
\begin{align*}
  \langle \alpha,\beta\rangle_{X}
  \leq \omega \Tiff \beta \leq \omega_{X}(\alpha)\,,\quad \text{for
    all 
    $\beta \in Q(\SX)$}\,.
\end{align*}
Suppose now that for each $\beta \in Q(\SX)$ we can find
$\Lneg{\beta} \in Q(X)$ such that
\begin{align*}
  \alpha & \leq \Lneg{\beta} \Tiff \langle \alpha,\beta\rangle_{X}
  \leq \omega \,,\quad \text{for
    all $\alpha\in Q(X)$}\,.
\end{align*}
This happens, for example, when $\brack{\funNB,\beta}_{X}$ has a right
adjoint and, in particular, when $Q$ monoidally factors through
$\SLatt$. Then,
the equivalences
\begin{align}
  \label{eq:Gconnection}
  \alpha & \leq \Lneg{\beta} \Tiff \langle \alpha,\beta\rangle_{X}
  \leq \omega \Tiff \beta \leq \omega_{X}(\alpha)\,,\quad \text{for
    all $\alpha\in Q(X)$ and $\beta \in Q(\SX)$}\,,
\end{align}
hold and exhibit $\omega_{X}$ and $\Lneg{\fun}$ as a Galois
connection. In particular $\Lneg{\beta}$ is uniquely determined.

\begin{theoremEnd}[category=staraut]{lemma}
  \label{lemma:omegainv}
  If 
  $\omega_{X} : Q(X)^{op} \rto Q(\SX)$ is invertible, then
  $\Lneg{\beta} = \omega^{-1}_{X} (\beta)$.
\end{theoremEnd}
\begin{proofEnd}
  We have $ \langle \alpha,\beta \rangle_{X} \leq \omega$ if and only
  if $\beta \leq \omega_{X}(\alpha)$, if and only if
  $\alpha \leq \omega_{X}^{-1}(\beta)$.
\end{proofEnd}

\begin{theoremEnd}[category=staraut]{lemma}
  \label{lemma:pairingTwist}
  If $\mu$ is natural, then  the following
  relation holds
  \begin{align}
    \label{eq:pairingTwist}
    \langle \alpha, \beta\rangle_{X} & = \langle \beta, Q(j_{X})(\alpha)\rangle_{\SX}\,,
  \end{align}
  for each $\alpha \in Q(X)$ and $\beta \in Q(\SX)$.  Consequently:
  \begin{enumerate}
  \item If
    $Q(j_{X})$ has a \ra $Q(j_{X})^{\ast}$, then $\Lneg{\beta} = Q(j_{X})^{\ast}(\omega_{\SX}(\beta))$.
  \item    If
    $j_{X}$ is invertible, then $\Lneg{\beta} = Q(j_{X}^{-1})(\omega_{\SX}(\beta))$.
  \end{enumerate}
\end{theoremEnd}
\begin{proofEnd}
  We obtain \eqref{eq:pairingTwist} by diagram chasing:
  \begin{center}
    \begin{tikzcd}[column sep=15mm,ampersand replacement=\&]
      Q(X) \times Q(\SX) \ar{d}{\mu_{X,\SX}}\ar{r}{\sigma_{X,
        \SX}} \&  Q(\SX) \times
      Q(X) \ar{d}{\mu_{\SX,X}}\ar{r}{\id \times Q(j_{X})} \& Q(\SX) \times Q(\SSX) \ar{d}{\mu_{\SX,\SSX}}\\
      Q(X \tensor \SX) \ar[swap]{rd}{Q(\ev_{X,0})}\ar{r}{Q(\sigma_{X,\SX})} \& Q(\SX \tensor
      X) \ar{r}{Q(\id\tensor j_{X})} \& Q(\SX \tensor \SSX) \ar{ld}{Q(\ev_{\SX,0})}\\
      \& 
      Q(0)
    \end{tikzcd}
  \end{center}
  Using \eqref{eq:pairingTwist}, we argue as follows:
  \begin{align*}
    \langle \alpha, \beta\rangle_{X} \leq \omega
    & \Tiff \langle \beta,Q(j_{X})(\alpha)\rangle_{X} \leq \omega \\
    &  \Tiff Q(j_{X})(\alpha) \leq \omega_{\SX}(\beta) \\
    & \Tiff \alpha \leq Q(j_{X})^{\ast}(\omega_{\SX}(\beta)) =
    Q(j_{X}^{-1})(\,\omega_{\SX}(\beta)\,)\,,
  \end{align*}
  where the equality
  $Q(j_{X})^{\ast}(\omega_{\SX}(\beta)) =
  Q(j_{X}^{-1})(\omega_{\SX}(\beta))$ holds when $j_{X}$ is
  invertible.
\end{proofEnd}

\begin{theoremEnd}[category=staraut]{theorem}
  \label{thm:omegaInv}
  If $\mu$ is natural, then, for each $\omega \in Q(0)$, $(0,\omega)$
  is dualizing if and only if $0$ is dualizing in $\C$ and, for each
  object $X$ of $\C$, $\omega_{X}$ is invertible.
\end{theoremEnd}
\begin{proofEnd}
  We only need to prove the converse of
  Lemma~\ref{lemma:easyDirect}, namely that
  diagram~\eqref{diag:staraut} commutes.
  \newline
  By Lemmas~\ref{lemma:omegainv} and \ref{lemma:pairingTwist}, we have
  $Q(j_{X}^{-1})(\omega_{\SX}(\beta)) = \Lneg{\beta} =
  \omega_{X}^{-1}(\beta)$, for each $\beta \in Q(\SX)$, thus
  $Q(j_{X}^{-1})\circ \omega_{\SX} = \omega_{X}^{-1}$, from which
  commutativity of \eqref{diag:staraut} follows.
\end{proofEnd}

\section{The double negation nucleus}
\label{sec:doubleNegation}

It is well known, see for instance \cite[Theorem~1]{Rosenthal1990a},
that given an element $\omega$ of a unital commutative quantale
$(Q,e,\qmult)$, the map $(\funNB\impl \omega)\impl \omega$ is a closure
operator, actually a nucleus. Taking the fixed points of this nucleus
yields a quantale $(Q_j,j(e),\qmult_j)$ in which every element is
equal to its double negation. Otherwise said, $(Q_j,j(e),\qmult_j)$ is
a \saut complete poset, also called a Girard quantale. In this section
we show how to reproduce this process with the total category in place
of the quantale.

We suppose from now on that $\C$ is \saut with $0$ as dualizing
element, and that $Q : \C \rto \SLatt$ is a monoidal functor, so that,
in particular, $\mu$ is natural. Thus, as stated in
Theorem~\ref{thm:SLattClosed}, $\iQ$ is \SMC. 
We fix $\omega \in Q(0)$. 

\begin{theoremEnd}[category=doubleNegation]{proposition}
  \label{prop:Qj}
  Define, for each $X \in \C$ and $\alpha \in Q(X)$,
  $\j_{X}(\alpha) \eqdef \Lneg{(\omega_{X}(\alpha))}$. Then $\j_{X}$
  is a closure operator on $Q(X)$. 
  For $X$ an object of $\C$ and $f : X \rto Y$ an arrow of $\C$, let
  \begin{align*}
    \Qj(X) & \eqdef \set{\alpha \in Q(X) \mid \j_{X}(\alpha) =
      \alpha}\,,
    & \Qj(f) & \eqdef \j_{Y} \circ Q(f)\,.
  \end{align*}
  Then 
  $\Qj : \C \rto \SLatt$ is a functor 
  and $\j_{X} : Q(X) \rto \Qj(X)$ is a surjective \sp map, natural in
  $X$.  If $f : X \rto Y$ is invertible, then $Q(f)$ restricts to a
  map from $\Qj(X) \rto \Qj(Y)$.
\end{theoremEnd}
\begin{proofEnd}
  In view of the relations in~\eqref{eq:Gconnection}, $\Lneg{(\cdot)}$
  and $\omega_{X}$ form a Galois connection and, by general theory,
  see e.g. \cite[Chapter 7]{DP}, their composition $\j_{X}$ is a
  closure operator on $Q(X)$. Then, $\Qj(X)$ is a complete lattice
  where the supremum of a family $\set{\alpha_{i} \mid i \in I}$ is
  $\j_{X}(\bigvee_{i} \alpha_{i})$ (we shall denote this supremum by
  $\supj_{i \in I}\alpha_{i}$).  Moreover
  $\j_{X} : Q(X) \rto \Qj(X)$ is \sp and surjective. 
  \newline
  It is obvious that $\Qj$ preserves the identities. The fact that it
  preserves the composition and the naturality of $\j$ are directly
  induced by the equality
  \begin{align}
    \label{eq:quotient}
    \j_{Y}(Q(f)(\j_{X}(\alpha))) & = \j_{Y}(Q(f)(\alpha))\,,
  \end{align}
  valid for each $\alpha \in Q(X)$.  This equality is proved by the
  following argument.  By Corollary~\ref{cor:omegaNatural}, for
  $f : X \rto Y$ in $\C$, the following diagram commutes:
  \begin{equation}
    \label{diag:omegaNatural}
    \begin{tikzcd}[ampersand replacement=\&]
      Q(X) \ar{d}{Q(f)}\ar{r}{\omega_{X}} \& Q(\SX)^{\op}  \ar{d}{Q(\Star{f})^{\op}}\\
      Q(Y) \ar{r}{\omega_{Y}} \& Q(\SY)^{\op} \,.
    \end{tikzcd}
  \end{equation}
  Recall now that if a diagram in \Pos
  \begin{equation}
    \label{diag:omegaNatural}
    \nonumber
    \begin{tikzcd}[ampersand replacement=\&]
      A  \ar{d}{h} \ar{r}{f} \& B  \ar{d}{g}\\
      C \ar{r}{k} \& D \,.
    \end{tikzcd}
  \end{equation}
  commutes and is such that $f$ and $k$ have right adjoints $\rho(f)$ and $\rho(k)$,
  respectively, then
  \begin{align}
    \label{eq:antidiagRadjs}
    h(\rho(f)(b)) & \leq \rho(k)(g(b))\,,
  \end{align}
  holds for each $b \in B$. The Galois connection expressed
  in~\eqref{eq:Gconnection} can also be understood by saying that
  $\Lneg{\fun} : Q(\SX)^{\op} \rto Q(X)$ is right adjoint to
  $\omega_{X} : Q(X) \rto Q(\SX)^{\op}$.
  \newline
  Thus, applying \eqref{eq:antidiagRadjs} to diagram
  \eqref{diag:omegaNatural}, we obtain
  \begin{align*}
    Q(f)(\Lneg{\beta}) & \leq \Lneg{(Q(\Star{f})^{\op}(\beta))} 
  \end{align*}
  for each $\beta \in Q(\SX)$.
  And therefore, by the naturality of $\omega$, 
  \begin{align}
    Q(f)(\j_{X}(\alpha)) &  = Q(f)(\Lneg{(\omega_{X}(\alpha))})
    \nonumber \\
    & \leq\Lneg{( (Q(\Star{f})^{\op})(\omega_{X}(\alpha)))}
    = \Lneg{(\omega_{Y}(Q(f)(\alpha)))} = \j_{Y}(Q(f)(\alpha))\,.
    \label{eq:jQLaxNatural}
  \end{align}
  As $\j_Y$ is a closure operator, we also obtain that 
  \begin{align*}
    \j_{Y}(Q(f)(\alpha)) & \leq  \j_{Y}(Q(f)(\j_{X}(\alpha)))
    \leq  \j_{Y}(\j_{Y}(Q(f)(\alpha)) =  \j_{Y}(Q(f)(\alpha))\,.
  \end{align*}
  That is, we have derived \eqref{eq:quotient}.
  This equality also allows us to prove that $\Qj(f)$ is sup-preserving:
  \begin{align*}
    \Qj(f)(\,\supj_{i} \alpha_{i}\,) & =
    \j_{Y}(Q(f)(\,\j_{X}(\bigvee_{i} \alpha_{i})\,))]
    = \j_{Y}(Q(f)(\bigvee_{i} \alpha_{i})) = \j_{Y}(\bigvee_{i}
    Q(f)(\alpha_{i})) \\
    & 
    \leq \j_{Y}(\bigvee_{i}
    \j_{Y}(Q(f)(\alpha_{i}))) = \supj_{i} \Qj(f)(\alpha_{i})\,,
  \end{align*}
  where the opposite inclusion follows from monotonicity of $\Qj(f)$.
  Finally, suppose $f : X \rto Y$ invertible. We aim at showing that
  $\j_{Y}(Q(f)(\alpha)) \leq Q(f)(\j_{X}(\alpha))$ (which then is an
  equality by equation~\eqref{eq:jQLaxNatural}).
  Indeed, we have
  \begin{align*}
    Q(f^{-1})(\j_{Y}(Q(f)(\alpha))) & \leq
    \j_{X}(Q(f^{-1})(Q(f)(\alpha))) = \j_{X}(\alpha)\,,
  \end{align*}
  and we obtain the desired inequality by applying $Q(f)$.
\end{proofEnd}

\medskip

Before arguing that $\Qj : \C \rto \SLatt$ is monoidal, we first give
a variant of Theorem~\ref{thm:omegaInv}, based on 
closure operators $\j_{X}$.
\begin{theoremEnd}[category=doubleNegation]{proposition}
  \label{prop:starautFromFixpoints}
  $(0,\omega)$ is a dualizing element of $\int Q$ if and only if $0$
  is a dualizing element of $\C$ and, for each object $X$ of $\C$ and
  each $\alpha \in Q(X)$, $\j_{X}(\alpha) =\alpha$.
\end{theoremEnd}
\begin{proofEnd}
  Recall that we assume that $\mu$ is natural. Thus, we rely on
  Theorem~\ref{thm:omegaInv} and argue that $\omega_{X}$ is invertible
  if and only if $\Qj(X) = Q(X)$, for each object $X$ of $\C$.
  \newline
  If $\omega_{X}$ is invertible, then it is inverted by its adjoint
  $\Lneg{\funNB}$, thus $\j_{X}(\alpha) = \alpha$ for each $\alpha \in
  Q(X)$.
  \newline
  Conversely, suppose that, for each object $X$ of $\C$,
  $\Qj(X) = Q(X)$. This amounts to saying that
  $\id_{Q(X)} = \j_{X} = Q(j_{X}^{-1}) \circ \omega_{\SX} \circ
  \omega_{X}$. Therefore, $\omega_{X}$ is split mono, and
  $\omega_{\SX}$ is split epi, for each for each object $X$ of
  $\C$. Therefore $\omega_{\SX}$ is both split mono and split epi,
  and therefore it is an an isomorphism. It follows that $\omega_{X}$ is
  an isomorphism as well.
\end{proofEnd}

\textEnd[category=doubleNegation]{
Let us develop a few remarks.  In the following, let
$\psi : (X \tensor Y \impl Z) \rto Y \impl (X \impl Z)$ be the
(natural) isomorphism determined by commutativity of
\begin{center}
  \begin{tikzcd}[column sep=20mm]
    X \tensor Y \tensor ((X \tensor Y) \impl Z) \ar[swap]{ddr}{\ev_{X
        \tensor Y,Z}}
    \ar{r}{X \tensor Y \tensor \psi}
    &     X \tensor Y \tensor (Y \impl (X \impl Z)) \ar{d}{X \tensor
      \ev_{Y, X \impl Z}} \\
    & X \tensor (X \impl Z) \ar{d}{\ev_{X,Z}} \\
    & Z
  \end{tikzcd}
\end{center}
}
\begin{theoremEnd}[category=doubleNegation]{lemma}
  The following relation holds:
  \begin{align}
    \label{eq:doubleImpl}    
    \langle\,\mu_{X,Y}(\alpha,\beta),\,\gamma\rangle_{X\tensor Y,Z}
    & = \langle \,\alpha,\,\langle \beta,Q(\psi)(\gamma)\,\rangle_{Y, X
      \impl Z}\,\rangle_{X,Z}\,,
  \end{align}
  for each $\alpha \in Q(X)$, $\beta \in Q(Y)$, and
  $\gamma \in Q((X \tensor Y)\impl Z)$.
\end{theoremEnd}
\begin{proofEnd}
  Recall that we assume in this Section that $\mu$ is natural.
  Equation~\eqref{eq:doubleImpl} is argued via the diagram chasing
  in Figure~\ref{fig:profAAA}.
  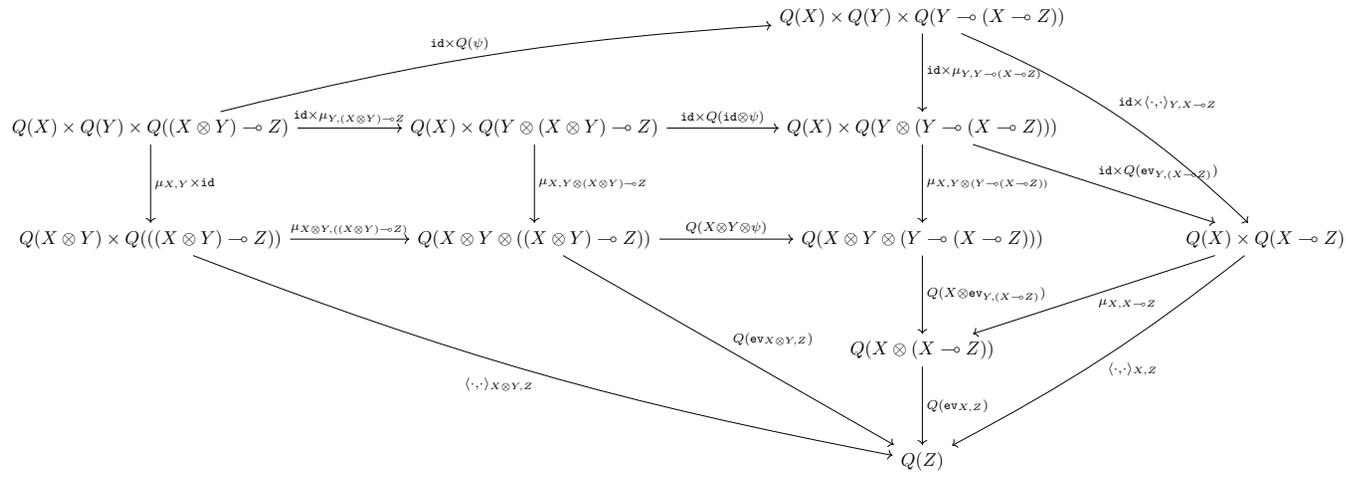
\begin{figure}
    \centering
    \adjustbox{rotate=-90,scale=0.70}{
      \begin{tikzcd}[column sep=20mm, row sep=15mm, ampersand replacement=\&]
      \& \&  Q(X) \times Q(Y)
      \times Q(Y \impl
      (X \impl Z)) 
      \ar{d}{\id\times\mu_{Y,Y \impl (X \impl Z)}}
      \ar[bend left=10]{rdd}{\id \times \langle \cdot, \cdot\rangle_{Y,X
        \impl Z}}
      \\ 
      Q(X) \times Q(Y) \times Q((X \tensor Y) \impl Z)
      \ar{d}{\mu_{X,Y} \times \id} \ar{r}{\id \times \mu_{Y,(X
            \tensor Y) \impl Z}} \ar[bend left=5]{rru}{\id \times Q(\psi)} \&
      Q(X) \times Q(Y \tensor (X \tensor Y) \impl Z) \ar{d}{\mu_{X, Y
          \tensor (X \tensor Y)\impl Z}} \ar{r}{\id \times Q(\id
        \tensor \psi)}\& Q(X) \times Q(Y \tensor (Y \impl (X \impl Z)))
      \ar{d}{\mu_{X,Y \tensor (Y \impl (X \impl Z))}} \ar{rd}{\id \times
        Q(\ev_{Y,(X \impl Z)})} \\ 
      Q(X \tensor Y) \times Q(((X \tensor Y) \impl Z)) \ar{r}{\mu_{X
          \tensor Y, ((X \tensor Y) \impl Z)}} \ar[swap,bend
      right=5]{rrdd}{\langle \cdot,\cdot\rangle_{X \tensor Y,Z}} \& Q(X
      \tensor Y \tensor ((X \tensor Y) \impl Z)) \ar{rdd}{Q(\ev_{X\tensor
          Y,Z})} \ar{r}{Q(X\tensor Y \tensor \psi) } \& Q(X\tensor Y \tensor (Y \impl (X \impl Z))) \ar{d}{Q(X
        \tensor \ev_{Y,(X \impl Z)})}  \&Q(X)
      \times Q(X \impl Z)
      \ar{ld}{\mu_{X,X \impl Z}} \ar[bend left=5]{ldd}{\langle \cdot,\cdot\rangle_{X,Z}}\\ 
       \&\&  Q(X \tensor (X \impl Z))\ar{d}{Q(\ev_{X,Z})}\\
       \&\&Q(Z)
    \end{tikzcd}
  }
    \caption{Proof of equation~\eqref{eq:doubleImpl}    
}
    \label{fig:profAAA}
  \end{figure}
\end{proofEnd}

\medskip

Recall that the double negation closure operator
$j\fun \eqdef (\funNB \impl \omega) \impl \omega$ on a \ucq
$(Q,e,\qmult)$ is a \emph{nucleus}, meaning that it satisfies the
following inclusion:
\begin{align*}
  j(j(\alpha) \qmult \beta) & \leq j(\alpha \qmult \beta)\,.
\end{align*}
The above inclusion is the key property allowing to construct a
quantale structure on the set $Q_{j}$ of fixed point of $j$. We show
that a similar property allows to transform $\Qj$ into a monoidal
functor.
\begin{theoremEnd}[category=doubleNegation]{proposition}
  The collection of maps $\j_{X}$ forms a nucleus w.r.t. $\mu$. That
  is, the following inequation holds:
  \begin{align}
    \label{eq:firstgoal}
    \mu_{X,Y}(\,\j_{X}(\alpha),\,\beta\,) & \leq \j_{X \tensor
      Y}(\,\mu_{X,Y}(\,\alpha,\,\beta\,)\,)\,,
  \end{align}
  for each $\alpha \in Q(X)$ and $\beta \in Q(Y)$. Consequently, by
  defining 
  $\uj \eqdef \j_{I}(u)$ and
  \begin{align*}
    \muj_{X,Y} & \eqdef \Qj(X) \times \Qj(Y) \rto[\mu_{X,Y}] Q(X
    \tensor Y) \rto[\j_{X \tensor Y}] \Qj(X \tensor Y)\,,
  \end{align*}
  these maps turn $\Qj : \C \rto \SLatt$ into a monoidal functor.
\end{theoremEnd}
\begin{proofEnd}
  Assume~\eqref{eq:firstgoal}.
  Then, by symmetry, we also have
  $$\mu_{X,Y}(\alpha,\j_{Y}(\beta)) \leq \j_{X \tensor
    Y}(\mu_{X,Y}(\alpha,\beta))\,.$$ Moreover, by monotonicity,
  $ \j_{X \tensor Y}(\mu_{X,Y}(\alpha,\beta)) \leq\j_{X \tensor
    Y}(\mu_{X,Y}(\j_{X}(\alpha),\j_{Y}(\beta))) $.  Altogether we
  deduce:
  \begin{align}
    \label{eq:maingoal}
    \j_{X \tensor Y}(\mu_{X,Y}(\j_{X}(\alpha),\j_{Y}(\beta))) & =
    \j_{X \tensor Y}(\mu_{X,Y}(\alpha,\beta))\,.
  \end{align}
  From \eqref{eq:maingoal} it immediately follows that $\muj$ is
  bilinear (as we assume that $\mu$ is bilinear) and that it makes
  commutative the diagrams required for monoidality, see
  Figure~\eqref{fig:monoidalFunctor}. Also, \eqref{eq:maingoal} yields
  naturality of $\muj$:
  \begin{align*}
    \Qj(f \tensor g)(\muj_{X,Y}(\alpha,\beta))
    & = \j_{X' \tesnor Y'}(Q(f \tensor
    g)(\j_{X,Y}(\mu_{X,Y}(\alpha,\beta)))) \\
    & = \j_{X' \tesnor Y'}(Q(f \tensor
    g)(\mu_{X,Y}(\alpha,\beta)))\,, 
    \tag*{by equation~\eqref{eq:quotient},}\\
    & = \j_{X' \tesnor Y'}(\mu_{X',Y'}(Q(f)(\alpha),Q(g)(\beta))),
    \tag*{since $\mu$ is natural,}\\
    & = \j_{X' \tesnor
      Y'}(\mu_{X',Y'}(\j_{X'}(Q(f)(\alpha)),\j_{Y'}(Q(g)(\beta))))\,,
    \tag*{by \eqref{eq:maingoal},}
    \\
    & = \muj_{X',Y'}(\Qj(f)(\alpha),\Qj(g)(\beta))\,.
  \end{align*}
  \newline
  Let us tackle the proof of relation \eqref{eq:firstgoal}, which is equivalent to
  \begin{align*}
    \langle
    \,\mu_{X,Y}(\j_{X}(\alpha),\beta),\,\omega_{X}(\mu_{X,Y}(\alpha,\beta))\,\rangle_{X
      \tensor Y} & \leq \omega\,.
  \end{align*}
  We derive the latter relation by means of the folllowing
  computations:
  \begin{align*}
    \langle
    \,\mu_{X,Y}(\j_{X}(\alpha),\beta),\,&\omega_{X}(\mu_{X,Y}(\alpha,\beta))\,\rangle_{X
      \tensor Y,0}\\
    & = \langle \j_{X}(\alpha), \langle \beta
    ,\,Q(\psi)(\omega_{X \tensor Y}(\mu_{X,Y}(\alpha,\beta)))\,\rangle_{Y,\SX}
    \rangle_{X,0}\,,
    \tag*{by equation \eqref{eq:doubleImpl},}
    \\
    & = \langle \j_{X}(\alpha), \langle \beta
    ,\iota_{Y,X^*}(\beta,\omega_{X}(\alpha))\rangle_{Y,\SX}
    \rangle_{X,0}\,,
    \tag*{by equation \eqref{eq:goal}, to be argued for,}
    \\
    \tag*{by equation \eqref{eq:counit0},}
    & \leq \langle \j_{X}(\alpha), \omega_{X}(\alpha) \rangle_{X,0}\,, \\
    & = \langle \Lneg{(\omega_{X}(\alpha))}, \omega_{X}(\alpha) \rangle_{X,0} \leq
    \omega\,.
  \end{align*}
  To complete the proof, we need to asses the equality
  \begin{align}
    \label{eq:goal}
    Q(\psi)(\omega_{X \tensor Y}(\mu_{X,Y}(\alpha,\beta))) & =
    \iota_{Y,X^*}(\beta,\omega_{X}(\alpha))\,.
  \end{align}
  The canonical isomorphism
  $\psi : (X\tensor Y) \impl Z \rto Y \impl (X \impl Z)$ lifts from
  $\C$ to $\iQ$, which means that the following diagram
  commutes:
  \begin{center}
    \begin{tikzcd}[column sep=30mm, ampersand replacement=\&]
      Q(X)^{op} \times Q(Y)^{op} \times Q(Z) \ar{d}{\mu_{X,Y}^{op}
        \times Q(Z)} \ar{r}{\sigma_{Q(X)^{op}, Q(Y)^{op}} \times Q(Z)}
      \& Q(Y)^{op} \times Q(X)^{op} \times Q(Z) \ar{d}{Q(Y)^{op}
        \times \iota_{X,Z}}
      \\
      Q(X \tensor Y)^{op} \times Q(Z) \ar{d}{\iota_{X\tensor Y, Z}}
      \& Q(Y)^{op} \times Q(X \impl Z) \ar{d}{\iota_{Y, X \impl Z}}\\
      Q((X \tensor Y) \impl Z) \ar{r}{Q(\psi)} \& Q(Y \impl (X \impl
      Z) )
      \end{tikzcd}
  \end{center}
  Equation~\eqref{eq:goal} is just an instance of this commutativity, with $Z = 0$.
\end{proofEnd}
\medskip
Next, the important feature of the \ucq $Q_{j}$ is that the
implication in $Q_{j}$ is computed exactly as in $Q$ and, moreover,
that $\omega \in Q_{j}$. We collect analogous remarks for the functor
$\Qj$.
\begin{theoremEnd}[category=doubleNegation]{lemma}
  \label{lemma:ImplLaxClosed}
  \mbox{\hspace{1mm}}\newline
  \begin{enumerate}
  \item The following inclusion holds:
    \begin{align}
      \label{eq:ImplLax}
      \j_{X \impl Y}(\iota_{X,Y}(\alpha,\beta)) & \leq
      \iota_{X,Y}(\j_{X}(\alpha),\j_{Y}(\beta))\,.
    \end{align}
    Consequently, if $\beta \in \Qj(Y)$, then
    $\iota_{X,Y}(\alpha,\beta) \in \Qj(X \impl Y)$, for each
    $\alpha \in Q(X)$.
\item    Every element of the form $\omega_{X}(\alpha)$, $\alpha \in Q(X)$,
  belongs to $\Qj(\SX)$.
\item   The element $\omega \in Q(0)$ is a fixed point of $\j_{0}$.
  \end{enumerate}
\end{theoremEnd}
\begin{proofEnd}
  \begin{enumerate}
  \item 
  We have
  \begin{align*}
    \langle \j_{X}(\alpha) , \j_{X \impl
      Y}(\iota_{X,Y}(\alpha,\beta))\rangle_{X,Y}
    & = Q(\ev_{X,Y})(\mu_{X,X\impl Y}(\j_{X}(\alpha),\j_{X \impl
      Y}(\iota_{X,Y}(\alpha,\beta)))) \\
    & \leq Q(\ev_{X,Y})(\j_{X \tensor (X \impl Y)}(\mu_{X,X\impl
      Y}(\alpha,\iota_{X,Y}(\alpha,\beta))))
    \tag*{using equation \eqref{eq:firstgoal}}
    \\
    & \leq \j_{Y}(Q(\ev_{X,Y})(\mu_{X,X\impl
      Y}(\alpha,\iota_{X,Y}(\alpha,\beta))))
    \tag*{by equation \eqref{eq:jQLaxNatural}}
    \\ 
    & \leq \j_{Y}(\beta)\,, \tag*{by equation~\eqref{eq:counit0}.}
  \end{align*}
  From the above inequality, equation \eqref{eq:ImplLax} follows using
  adjointness.  Next, assume $\beta \in \Qj(Y)$, let
  $\alpha \in Q(X)$, and observe that
  \begin{align*}
    \iota_{X,Y}(\alpha,\beta) & \leq \j_{X \impl
      Y}(\iota_{X,Y}(\alpha,\beta))\\
    & \leq \iota_{X,Y}(\j_{X}(\alpha),\j_{Y}(\beta))\,, \tag*{by
      \eqref{eq:ImplLax},}
    \\
    & = \iota_{X,Y}(\j_{X}(\alpha),\beta)\,,
    \tag*{since $\beta \in \Qj(Y)$,}\\
    & \leq \iota_{X,Y}(\alpha,\beta) \,, \tag*{since
      $\iota$ reverses the order in its first variable,}
  \end{align*}
  and therefore
  $\iota_{X,Y}(\alpha,\beta) = \j_{X \impl
    Y}(\iota_{X,Y}(\alpha,\beta))$.
  \item
  By Corollary~\ref{cor:omegaNatural}, the following diagram commutes:
  \begin{center}
    \begin{tikzcd}[ampersand replacement=\&]
      Q(X) \ar{d}{Q(j_{X})} \ar{r}{\omega_{X}} \& Q(\SX)^{op}
      \ar{d}{Q(\Star{(j_{X})})^{op} \,=\, Q(\Star{(j_{X})})^{-1} \,=\, Q(\Star{(j_{X})}{}^{-1})  \,=\, Q(j_{\SX})} \\
      Q(\SSX) \ar{r}{\omega_{\SSX}} \& Q(\SSSX)^{op} 
    \end{tikzcd}
  \end{center}
  where the equality $\Star{(j_{X})}{}^{-1} = j_{\SX}$ follows from
  the relation $\Star{(j_{X})}\circ j_{\SX} = \id_{\SX}$, valid in
  every \SMCC, and since $\Star{(j_{X})}$ is invertible whence mono.
  Using this observation, we can compute as follows:
  \begin{align*}
    \j_{\SX}\circ \omega_{X} & = Q(j_{\SX}^{-1}) \circ \omega_{\SSX}
    \circ \omega_{\SX} \circ \omega_{X}
    =  \omega_{X }\circ Q(j_{X}^{-1})\circ \omega_{\SX} \circ
    \omega_{X}
    = \omega_{X} \circ \j_{X} = \omega_{X}\,.
  \end{align*}
\item Since $\mu$ is monoidal, $Q(I)$ is a monoid in $\SLatt$ (that is, a
  quantale) and, for each object $X$ of $\C$, $Q(X)$ is a
  $Q(I)$-module. For $q \in Q(I)$ and $\alpha \in Q(X)$, the action is
  given by $q \cdot \alpha = Q(\lambda_{X})(\mu_{I,X}(q,\alpha))$.
  \newline
  Let $\lambda_{0}^{\sharp} : 0 \rto I \impl 0$ the transpose of
  $\lambda_{0} : I \tensor 0 \rto 0$ and recall $\lambda_{0}^{\sharp}$
  is an isomorphsm.
  Commutativity of the following diagram
  \begin{center}
    \begin{tikzcd}[column sep=20mm, ampersand replacement=\&]
      Q(I) \times Q(0) \ar{dd}{\mu_{I,0}}\ar{r}{Q(I) \times Q(\lambda_{0}^{\sharp})} \& Q(I) \times
      Q(I\impl 0) \ar{d}{\mu_{I,I \impl 0}}\\
       \& Q(I \tensor (I \impl 0)) \ar{d}{Q(\ev_{I,0})}\\
      Q(I \tesnor 0) \ar{r}{Q(\lambda_{0})} \& Q(0)
    \end{tikzcd}
  \end{center}
  yields the relation
  \begin{align*}
    q \cdot \alpha & = \brack{q,Q(\lambda_{0}^{\sharp})(\alpha)}_{I,0}\,.
  \end{align*}
  Let next $e$ be the unit of $Q(I)$ and compute as follows:
  \begin{align*}
    Q(\lambda_{0}^{\sharp})(\omega) & =
    Q(\lambda_{0}^{\sharp})(\bigvee \set{\beta \in Q(0) \mid e \cdot
      \beta \leq \omega }) \\
    & = \bigvee \set{Q(\lambda_{0}^{\sharp})(\beta) \mid \beta \in
      Q(0),\, \brack{e,Q(\lambda_{0}^{\sharp})(\beta)}_{I,0} \leq
      \omega
    } \\
    & = \bigvee \set{\gamma \in Q(I \impl 0) \mid
      \brack{e,\gamma}_{I,0} \leq \omega } = \omega_{I}(e)\,.
  \end{align*}
  It follows that
  $\omega = Q((\lambda_{0}^{\sharp})^{-1})(\omega_{I}(e))$.  By
  item~2, $\omega_{I}(e) \in \Qj(I \impl 0)$ and since morpshims of
  the form $Q(f)$ with $f : X \rto Y$ invertible restrict to maps from
  $\Qj(X) \rto \Qj(Y)$, see Proposition~\ref{prop:Qj},
  $\omega = Q((\lambda_{0}^{\sharp})^{-1})(\omega_{I}(e)) \in \Qj(0)$.
  \qedhere
\end{enumerate}
\end{proofEnd}

\medskip

The following theorem is 
a direct consequence of these
observations and of Proposition~\ref{prop:starautFromFixpoints}.
\begin{theoremEnd}[category=doubleNegation]{theorem}
  Let $\C$ be a \saut category and $Q : \C \rto \SLatt$ be a monoidal
  functor. Let $0$ be a dualizing object of $\C$ and pick
  $\omega \in Q(0)$. Then $(0,\omega)$ is a dualizing object of
  $\int \Qj$.
\end{theoremEnd}
\begin{proofEnd}
  As we are dealing with two monoidal functors, $Q$ and $\Qj$, we
  shall distinguish the structures making $\iQ$ and $\int \Qj$ by
  adding $\j$ in superscript. For example, we have the following
  characetrisation if $\brackj{\funNB,\funNB}_{X,Y}$: 
  \begin{align*}
    \brackj{\alpha,\beta}_{X,Y} & = \Qj(\ev_{X,Y})(\,\muj(\alpha,\beta)\,)
    \\
    & =  \j_{Y}(\,Q(\ev_{X,Y})(\j_{X\tensor (X \impl Y)}(\mu_{X,X \impl
      Y}(\alpha,\beta)))\,) \\
    & = \j_{Y}(Q(\ev_{X,Y})(\mu_{X,X \impl Y}(\alpha,\beta)))
        \tag*{by equation~\eqref{eq:quotient}}
    \\
    & = \j_{Y}(\brack{\alpha,\beta}_{X,Y})
    \,.
  \end{align*}
  Using this, we see that the right adjoint to $\brackj{\alpha,\cdot}$
  is $\iota_{X,Y}(\alpha,\cdot)$. Indeed, for $\alpha \in \Qj(X),
  \beta \in \Qj(X \impl Y)$, and $\gamma \in \Qj(Y)$, we have
  $\iota_{X,Y}(\alpha,\gamma) \in \Qj(Y)$ and
  \begin{align*}
    \brackj{\alpha,\beta}_{X,Y} = \j_{Y}(\brack{\alpha,\beta}_{X,Y}) \leq \gamma
    & \Tiff
    \brack{\alpha,\beta}_{X,Y} \leq \gamma  \Tiff
    \beta \leq \iota_{X,Y}(\alpha,\gamma)\,.
  \end{align*}
  Since $\omega \in \Qj(0)$, then
  $\omegaj_{X}(\alpha) = \iota_{X,0}(\alpha,\omega) =
  \omega_{X}(\alpha)$.  We also have $\Lnegj{\beta} =
  \Lneg{\beta}$. Indeed, $\beta \in \Qj(\SX)$, since
  $\Lneg{\beta} = Q(j^{-1}_{X})(\omega_{\SX}(\beta)) \in \Qj(X)$ by
  Lemma~\ref{lemma:pairingTwist}, $\omega_{\SX}(\beta) \in \Qj(\SSX)$
  by item 2 of Lemma~\ref{lemma:ImplLaxClosed}, $Q(j^{-1}_{X})$
  restricts to a map $\Qj(\SSX) \rto \Qj(X)$ by
  Proposition~\ref{prop:Qj}.
  \newline
  Thus
  $\Lnegj{(\omegaj_{X}(\alpha))} = \Lneg{(\omega_{X}(\alpha))}=
  \j_{X}(\alpha) = \alpha$, for every object $X$ of $\C$ and every
  $\alpha \in \Qj(X)$, and therefore $\int \Qj$ is \saut, using
  Proposition~\ref{prop:starautFromFixpoints}.
\end{proofEnd}

\medskip

\newcommand{\Puq}{P \circ UQ} 
\newcommand{\PUQ}{(P \circ UQ)} 
\newcommand{\jindep}{\jmath^{\independent}} 
\newcommand{\PUQj}{\PUQ^{\jindep}}

We end this section by giving a representation theorem analogous to
the well-known representation theorem for Girard quantales that yields
the completeness of phase semantics of linear logic. The theorem is
also meant to illustrate how our construction relate to the focused
orthogonality categories introduced in \cite{HS2003} and further
studied in \cite{Fiore2023}, see also Example~\ref{example:focs}.
\begin{theoremEnd}[category=doubleNegation]{theorem}
  \label{thm:representation} 
  Let $Q : \C \rto \SLatt$ be a monoidal functor and suppose that
  $(0,\omega)$ is a dualizing element of $\iQ$. Consider the functor
  $P \circ UQ$, where $UQ$ is the composition
  $\C \rto[Q] \SLatt \rto[U] \Set$ of $Q$ and the forgetful functor to
  $\Set$, and where $P : \Set \rto \SLatt$ is the covariant powerset
  functor (free complete sup-lattice functor). Define
  \begin{align*}
    \independent & \eqdef \set{ x \in Q(0) \mid x \leq \omega}\,.
  \end{align*}
  Then, the maps sending $x \in Q(X)$ to
  $\Downset{x} \eqdef \set{x' \in{ Q(X) \mid x' \leq x}} \subseteq
  Q(X)$ form a natural isomorphism between the functors
  $Q : \C \rto \SLatt$ and
  $(P \circ UQ)^{\jmath^{\independent}} : \C \rto \SLatt$.
\end{theoremEnd}
\begin{proofEnd}
  Let us argue that a set $A \subseteq Q(X)$ belongs to $\PUQj(X)$ if
  and only if it is of the form $\Downset{\alpha}$ for some
  $\alpha \in Q(X)$. To this goal, observe that, for $B \in \PUQ(\SX)$,
  \begin{align*}
    \Lneg{B}  & =  \set{ \alpha \in Q(X) \mid \brack{\alpha,\beta}
      \leq \omega,\, \text{for all $\beta \in B$}}\,.
  \end{align*}
  Moreover, by general theory of Galois connections, all the elements
  of $\PUQj(X)$ are of this form.  Now, it is immediate that if
  $\alpha \in \Lneg{B}$ and $\alpha' \leq \alpha$, then
  $\alpha' \in \Lneg{B}$. Moreover, since $\brack{\funNB, \beta}_{X}$
  is \sp, $\bigvee (\Lneg{B}) \in \Lneg{B}$. Thus, if
  $A \in \PUQj(X)$, then $A = \Downset{\bigvee A}$.
  \newline
  Next, for $\alpha \in Q(X)$, $\Downset{\alpha} \in \PUQj(X)$, since
  \begin{align*}
    \Lneg{\set{\omega_{X}(\alpha)}} & = \set{\alpha' \in Q(X) \mid
      \brack{\alpha',\omega_{X}(\alpha)}_{X} \leq \omega} 
     = \Downset{\j_{X}(\alpha)} = \Downset{\alpha}\,,
  \end{align*}
  since we are assume that $Q(X) = \Qj(X)$.  Thus, $\PUQj(X)$ is the
  set of all the principal downsets of $Q(X)$ and
  $\Downset{\funNB} : Q(X) \rto \PUQj(X)$ is an order isomorphism, thus
  an isomorphism in $\SLatt$.
  \newline
  For any closure operator $j$
  on a complete lattice $L$, $j(x)$ is the least element of
  $\set{ y \in L \mid j(y) = y}$. It follows then that
  \begin{align*}
    \jindep_{X}(A) & = \Downset{\bigvee A}\,, \quad \text{for each
      $A \subseteq Q(X)$}.
  \end{align*}
  For $f : X \rto Y$ in $\C$, we have
  \begin{align*}
    \PUQj(f)(\Downset{\alpha} ) & =  \Downset[Y]{\bigvee \set{Q(f)(\alpha')
      \mid \alpha' \leq \alpha}} \\
    & =  \Downset[Y] {Q(f)(\bigvee \set{\alpha'
      \mid \alpha' \leq \alpha}} = \Downset[Y]{Q(f)(\alpha)}\,,
  \end{align*}
  showing that the maps sending $\alpha \in Q(X)$ to
  $\Downset{\alpha} \in \PUQj(X)$ are natural in $X$.
\end{proofEnd}

\section{Lifting fixed points of functors}
\label{sec:liftingfixedpoints}

This research was also motivated by ongoing research on the algebraic
and categorical semantics of linear logic with fixed points
\cite{EJ2021,AdJS2022,Farzad2023}. In these works, several questions
about initial and \final (co)algebras of endofunctors of the category
$\Nuts$ were raised.  We shall argue in Example~\ref{exp:nuts} that
this category also arises as $\int Q$ for some functor $Q$ into
$\SLatt$.  This was realised in \cite{Fiore2023} where these questions
were also settled. Nonetheless, we believe pertinent to pinpoint here
how this kind of results relate to the general theory developed in
Section~\ref{sec:liftingFunctors}.

Recall that, given a category $\C$ and an endofunctor $F:\C \rto \C$,
an $F$-\emph{algebra} (resp. $F$-\emph{coalgebra}) is an object $X$
together with a morphism (called structure) $\gamma : F(X) \rto X$
(resp. $\gamma : X \rto F(X)$). A morphism of algebras
(resp. coalgebras) $(X,\gamma) \rto (Y,\delta)$ is simply a morphism
$f: X \rto Y$ such that the following diagram commutes (respectively):
\begin{center}\begin{tikzcd}
F(X) \arrow[r,"F(f)"] \arrow[d,"\gamma",swap] & F(Y) \arrow[d,"\delta"] & F(X) \arrow[r,"F(f)"]  & F(Y) \\
X \arrow[r,"f",swap] & Y & X \arrow[u,"\gamma"] \arrow[r,"f",swap] & Y \arrow[u,"\delta",swap]
\end{tikzcd}\end{center}
Algebras (resp. coalgebras) and their morphisms form a category
$\alg_C(F)$ (resp. $\coalg_C(F)$). Notice that we have
$\coalg_C(F)={\alg_{C^{\op}}(F^{\op})}^{\op}$, where
$F^{\op} : \C^{\op} \rto \C^{\op}$ is the same functor as $F$, just
considered as an endofunctor of $\C^{\op}$.

An \emph{initial algebra} (resp. \emph{\final coalgebra}) is an
initial (resp. \final) object of $\alg_C(F)$ ($\coalg_C(F)$). It is
well-known that the structure of an initial algebra (terminal
coalgebra) is invertible \cite{Lambek68}. For this reason, initial and
terminal (co)algebras of functors have been used to give semantics (of
proofs) to logics with least and greatest fixed point-operators.

\medskip

Let $Q:\C \rto \Pos$ and $F:\C\rto \C$ be functors and suppose that
there is a lifting $\lF : \int Q \rto \int Q$. Denote by
$\psi : Q \rto QF$ the lax natural transformation corresponding to
$\lF$, as given in Theorem~\ref{thm:lifting}.

\begin{theoremEnd}[category=liftingfixedpoints]{proposition}
  \label{coalgSQ'}
  For an F-coalgebra $(X,\gamma)$, define
  \begin{align*}
    \Qnu(X,\gamma) & \eqdef \set{ \alpha \in Q(X) \mid Q(\gamma)(\alpha) \leq
      \psi_X(\alpha) }\,.
  \end{align*}
  Then, for $f:(X,\gamma) \rto (Y,\delta)$ a coalgebra morphism,
  $Q(f)$ restricts from $\Qnu(X,\gamma)$ to $\Qnu(Y,\delta)$. Thus $\Qnu$ is a
  functor $\coalg_{\C}(F) \rto \Pos$ and   we have an isomorphism
  \begin{align}
    \label{eq:isocoagint}
    \coalg_{\iQ}(\lF) &\simeq \int \Qnu\,.
  \end{align}
\end{theoremEnd}
\begin{proofEnd}
  Let $f:(X,\gamma) \rto (Y,\delta)$ be a coalgebra morphism and
  assume $Q(\gamma)(\alpha) \leq \psi_{X}(\alpha)$.  Verification that
  $Q(\delta)(Q(f)(\alpha)) \leq \psi_{Y}(Q(f)(\alpha)$ is achieved as
  follows:
  \begin{align*}
    Q(\delta)(Q(f)(\alpha)) & = Q(Ff)(Q(\gamma)(\alpha))\,, \tag*{since
      $f$ is a coalgebra morphism,} \\
    & \leq Q(Ff)(\psi_{X}(\alpha))\,, \tag*{by assumption on $\alpha$,} \\
    & \leq (\psi_{Y}(Q(Ff)(\alpha))\,, \tag*{by lax-naturality of  $\psi$.} \\
  \end{align*}
  For $(X,\alpha) \in \int Q$, recall that
  $\lF(X,\alpha)=(F(X),\psi_X(\alpha))$. Therefore, an $\lF$-coalgebra
  $((X,\alpha),\gamma)$ is given by an $F$-coalgebra $(X,\gamma)$ such
  that
  \[
    Q(\gamma)(\alpha) \leq \psi_X(\alpha)\,.
  \]
  A morphism of $\lF$-coalgebras
  $f:((X,\alpha),\gamma) \rto ((Y,\beta),\delta)$ is just a morphism
  of $F$-coalgebra $f:X \rto Y$ in $\C$ such that
  $Q(f)(\alpha) \leq \beta$.
  \newline
  Then, the mapping sending an object $((X,\alpha),\gamma)$ of
  $\coalg_{\iQ}(\lF)$ to the object $((X,\gamma),\alpha)$ of $\int \Qnu$
  (and a map $f$ to itself) is the desired isomorphism
  $\coalg_{\iQ}(\lF) \simeq \int \Qnu$.
\end{proofEnd}

We can give a similar result for algebras if we assume that
$\psi : Q \rto Q\circ F$ is natural
or that $Q$ factors through $\SLatt$. 
For the latter case, we need to develop some tools, illustrating the
intrinsic dualities of the category $\SLatt$.  Suppose therefore that
$Q$ factors through $\SLatt$, so $Q(X)$ is a complete lattice and
$Q(f)$ has a \ra $Q(f)^{\ast}$. Denote by $\QS:C^{\op} \rto \Pos$ the
functor such that $\QS(X)\eqdef Q(X)^{\op}$ and
$\QS(f)\eqdef\Star{Q(f)}$. Obviously, also $\QS$ factors through $\SLatt$.
\begin{theoremEnd}[category=liftingfixedpoints]{lemma}
  \label{SQopSQ*}
  If $Q$ factors through $\SLatt$, then the following statements hold:
\begin{enumerate}
\item The functor $\Qnu$ factors through $\SLatt$. 
\item We have $(\int Q)^{\op} = \int \QS$ and the functor $\lF^{\op}$
  is a lifting of $F^{\op}$ to $\int \QS$. 
\end{enumerate}
\end{theoremEnd}
\begin{proofEnd}
  \begin{enumerate}
  \item First, observe that for $(b_i)_i \subset \Qnu(B,\beta)$---that
    is $b_i \in Q(B)$ is such that $Q(\beta)(b_i) \leq \psi_B(b_i)$,
    for each $i$---we have
    \[
      Q(\beta)(\bigvee_{i} b_i) \leq \bigvee_{i} Q(\beta)(b_i) \leq \bigvee_{i} \psi_B(b_i) \leq \psi_B(\bigvee b_i)
    \]
    because $Q(\beta)$ is \sp and $\psi_B$ is \opr. This shows that
    $\Qnu(B,\beta)$ is a subset of $Q(B)$ closed under
    suprema. Necessarily, for a coalgebra morphism
    $f : (B,\delta) \rto (A,\gamma)$, the restriction $Q(f)$ to
    $Q(B,\delta)$ preserves these suprema.  So $\Qnu(B,\beta)$ is a
    complete lattice and $\Qnu$ is functor because it is the
    restriction of $Q$ for morphism.
    \\
  \item Objects of $\int \QS$ are of the form $(X,\alpha)$ where
    $X \in \C$ and $\alpha \in Q(X)^{\op}$, where we can simply write
    $\alpha \in Q(X)$.  Morphisms $(X,\alpha) \rto (Y,\beta)$ are maps
    $f:Y \rto X$ in $\C$ such that $\Star{Q(f)}(\alpha) \leq \beta$ in
    $Q(Y)^{\op}$, that is, $\beta \leq \Star{Q(f)}(\alpha)$ in
    $Q(Y)$. Under the adjunction, this is equivalent to
    $Q(f)(\beta) \leq \alpha$ in $Q(X)$.  Let us spell out what is
    $(\iQ)^{\op}$. Objects of $(\iQ)^{\op}$ are pairs $(X,\alpha)$
    where $X \in C$ and $\alpha \in Q(X)$. A morphism
    $(X,\alpha) \rto (Y,\beta)$ in $(\iQ)^{\op}$ is a morphism
    $f:Y \rto X$ in $C$ such that $Q(f)(\beta) \leq \alpha$.  Moreover
    the following diagram commutes, thus showing the very last statement:
    \begin{equation*}
      \begin{tikzcd}[ampersand replacement=\&]
        \iQ^* \arrow[r,"\lF^{\op}"] \arrow[d,"\pi",swap] \& \iQ^*  \arrow[d,"\pi"]  \\
        C^{\op} \arrow[r,"F^{\op}",swap] \& C^{\op} \\
      \end{tikzcd}
      \tag*{\qedhere}
    \end{equation*}
  \end{enumerate}
\end{proofEnd}

\begin{theoremEnd}[category=liftingfixedpoints]{proposition}
  \label{algSQ'}
  Assume that $\psi : Q \rto Q \circ F$ is
  natural or that $Q$ factors through $\SLatt$.
  For an $F$-algebra $(X,\gamma)$, define
  \begin{align*}
    \Qmu(X,\gamma) & \eqdef \set{ \alpha \in Q(X) \mid Q(\gamma)(\psi_{X}(\alpha)) \leq
      \alpha }\,.
  \end{align*}
  If $\psi$ is natural, then, for $f:(X,\gamma) \rto (Y,\delta)$ an
  algebra morphism, $Q(f)$ restricts to a map from $\Qmu(X,\gamma)$ to
  $\Qmu(Y,\delta)$.
  If $Q$ factors through $\SLatt$, then $\Qmu$ is extended to a
  functor so to have $\Qmu = \QSnuS$.  In both cases, $\Qmu$ is made
  into a functor $\alg_{\C}(F) \rto \Pos$ and we have an isomorphism
  \begin{align}
    \label{eq:isoagint}
    \alg_{\iQ}(\lF) & \simeq \int \Qmu\,.
  \end{align}
\end{theoremEnd}
\begin{proofEnd}
  Suppose that $\psi$ is natural, let $f:(X,\gamma) \rto (Y,\delta)$
  be an algebra morphism and assume
  $Q(\gamma)(\psi_{X}(\alpha)) \leq \alpha$.  Verification that
  $Q(f)(\alpha) \in \Qmu(Y,\delta)$, that is,
  $Q(\delta)(\psi_{Y}(Q(f)(\alpha))) \leq Q(f)(\alpha)$, is achieved as
  follows:
  \begin{align*}
    Q(\delta)(\psi_{Y}(Q(f)(\alpha))) & =
    Q(\delta)(Q(Ff)(\psi_{X}(\alpha)))\,, \tag*{by naturality of
      $\psi$,} \\
    & = Q(f)(Q(\gamma)(\psi_{X}(\alpha)))\,, \tag*{since
      $f$ is an algebra morphism,} \\
    & \leq Q(f)(\alpha)\,, \tag*{by assumption on $\alpha$.}
  \end{align*}
  An $\lF$-algebra
  $((X,\alpha),\gamma)$ is given by an $F$-algebra $(X,\gamma)$ such
  that
  \[
    Q(\gamma)(\psi_{X}(\alpha)) \leq \alpha\,.
  \]
  A morphism of $\lF$-algebras
  $f:((X,\alpha),\gamma) \rto ((Y,\beta),\delta)$ is just a morphism
  of $F$-algebras $f:(X,\gamma) \rto (Y,\delta)$ in $\C$ such that
  $Q(f)(\alpha) \leq \beta$.
  Then, the mapping sending an object $((X,\alpha),\gamma)$ of
  $\alg_{\iQ}(\lF)$ to the object $((X,\gamma),\alpha)$ of $\int \Qmu$
  (and a map $f$ to itself) is the desired isomorphism
  $\alg_{\iQ}(\lF) \simeq \int \Qmu$.
  \medskip\newline
  Next, suppose that $Q$ factors through $\SLatt$. As we no more
  assume that $\psi$ is natural, for an algebra morphism
  $f : (X,\gamma) \rto (Y,\delta)$ and $\alpha \in \Qmu(X,\gamma)$,
  $Q(f)(\alpha)$ might not belong to $\Qmu(Y,\delta)$. Yet, since
  $\Qmu(Y,\delta)$ is a subset of $Q(Y)$ closed under infima, we can define
  \begin{align}
    \label{eq:defQmuf}
    \Qmu(f)(\alpha) & \eqdef \bigwedge \set{\beta \in \Qmu(Y,\delta)
      \mid Q(f)(\alpha) \leq  \beta}\,.
  \end{align}
  The definition in \eqref{eq:defQmuf} doesn't make explicit that
  $\Qmu$ is a functor.  We claim that this is the case, since we
  actually have the equality
  \begin{align}
    \label{eq:Qmu}
    \Qmu & = \QSnuS\,.
  \end{align}
  Assuming this equality, and recalling we can describe algebras in
  term of coalgebras, we
  obtain
  \begin{align*}
    \alg_{\int Q}(\lF) & = \left [\coalg_{(\int
        Q)^{\op}}(\lF^{\op}) \right ]^{\op} =
    \left [\coalg_{\int Q^*}(\overline{F^{\op}}) \right ]^{\op}\\
    & \simeq \left [ \int \QSnu \right ]^{\op} = \int \QSnuS = \int \Qmu\,,
  \end{align*}
  where we have used, in the order,  Lemma~\ref{SQopSQ*}, 
  Proposition~\ref{coalgSQ'},  again Lemma \ref{SQopSQ*}, and then \eqref{eq:Qmu}.
  \medskip\newline
  Next, let us show that $\Qmu=\QSnuS$. We have
  $\QSnu: \coalg_{\C^{\op}}(F^{\op}) \rto \Pos$, so
  $\QSnuS : \left [\coalg_{\C^{\op}}(F^{\op})\right ]^{\op} = \alg_\C(F)
  \rto \Pos$.
  Let $(X,\gamma) \in \alg_\C(F)$, we have
  $\QSnuS(X,\gamma)= \QSnu(X,\gamma)^{\op}$ and
  \begin{align*}
    \QSnu(X,\gamma) &= \set{\alpha \in Q(X)^{\op} \mid
      \QS(\gamma)(\alpha) \leq \psi_X(\alpha) \text{ in }QF(X)^{\op}
    }\,.
  \end{align*}
  Because the order of this poset is the one from $Q(X)^{\op}$,
  the order of the opposite poset is the one from
  $Q(X)$.  
  Also the inequality condition above is equivalent, by transposition,
  to $Q(\gamma)(\psi_X(\alpha)) \leq \alpha$. Then
  \begin{align*}
    \QSnuS(X,\gamma) & = \QSnu(X,\gamma)^{\op} = \set{\alpha \in Q(X)
      \mid Q(\gamma)(\psi_X(\alpha)) \leq \alpha \text{ in }Q(X) } =
    \Qmu(X,\gamma)\,.
  \end{align*}
  That is, the equality \eqref{eq:Qmu} holds on objects.
  \medskip\newline
  For morphisms, consider that $f:(X,\gamma) \rto (Y,\delta)$ in
  $\alg_\C(F)$.  Recall that, as a morphism of
  $\coalg_{\C^{\op}}(F^{\op})$, $f$ is typed the opposite way,
  $f:(Y,\delta) \rto (X,\gamma)$. 
  Let
  $\alpha \in \Qmu(X,\gamma)$, $\beta \in \Qmu(Y,\delta)$, and observe
  the following chain of equivalences:
  \begin{align*}
    \Qmu(f)(\alpha) \leq \beta \within{\Qmu(Y,\delta)} & \Tiff
    Q(f)(\alpha) \leq \beta
    \within[\,,]{Q(Y)}
    \tag*{by definition of $\Qmu(f)(\alpha)$,}
    \\
    & \Tiff   \alpha \leq Q(f)^{\ast}(\beta)
    \within{Q(X)} \\
    & \Tiff   \QS(f)(\beta) = Q(f)^{\ast}(\beta) \leq \alpha 
    \within{\QS(X) = Q(X)^{\op}} \\
    & \Tiff   \QSnu(f)(\beta) \leq \alpha 
    \within{\QSnu(X,\gamma)}
    \tag*{since $\alpha,\beta \in \QSnu(X) = \Qmu(X,\gamma)^{\op}$ and
      $\QS(f)$ restricts from $\QS$ to $\QSnu$} \\
    & \Tiff \beta \leq \QSnu(f)^{\ast}(\beta)
    \within{\QSnu(Y,\delta) } \\
    & \Tiff \QSnuS(f)(\alpha) = \QSnu(f)^{\ast}(\alpha) \leq \beta
  \end{align*}
  where the last inequality is taken in
  $\QSnu(Y,\delta)^{\op}= \QSnuS(Y,\delta) = \Qmu(Y,\delta)$.  From
  these equivalences we obtain the equality
  $\Qmu(f)(\alpha) = \QSnuS(f)(\alpha)$, for each
  $\alpha \in \Qmu(X,\gamma)$, as aimed.
\end{proofEnd}

The isomorphisms in~\eqref{eq:isocoagint} and~\eqref{eq:isoagint}
allow to reduce the existence (or lifting) of a \final coalgebra
(initial algebra) to the existence of a \final (initial) object in a
generic op-fibration $\int Q \rto \C$.  It is straightforward that if
$1$ a \final object of $\C$ and $\top_{1}$ is the supremum of $Q(1)$,
then $(1,\top_{1})$ is a \final object of $\iQ$.  Similarly, if $0$ is
an initial object of $\C$, $\bot_{0}$ is the least element of $Q(0)$,
and for each object $X$, the map $Q(!_{X}) : Q(0) \rto Q(X)$ sends
$\bot_{0}$ to a least element of $Q(X)$, then $(0,\bot_{0})$ is an
initial object of $\int Q$. Thus we have the following proposition, that
for convenience we state directly for functors $Q$ factoring through
$\SLatt$, leaving the possible generalizations to the reader.

\begin{theoremEnd}[category=liftingfixedpoints]{proposition}
  Let $Q : \C \rto \SLatt$ and suppose that an endofunctor $F$ of $\C$
  has a lifting $\lF$ to $\int Q$. 
  \begin{enumerate}
  \item If $F$ has a \final $F$-coalgebra $(\nu.F,\chi)$, 
    then the object $(\nu.F,\chi,\nu.\phi)$ of
    $\int \Qnu \simeq \coalg_{\iQ}(\lF)$ is \final, where $\nu.\phi$
    is the greatest fixed point of the map $\phi$ defined as the
    composition
    $Q(\nu.F) \rto[\psi_{\nu.F}] Q(F(\nu.F)) \rto[Q(\chi^{-1})]
    Q(\nu.F)$.
  \item Suppose that $\psi: Q \rto Q \circ F$ is natural and that $F$ has an
    initial algebra $(\mu.F,\chi)$. 
    For each $F$-algebra $(X,\gamma)$,
    let $\mu.\phi_{X}$ denote the least fixed point of the map
    $\phi_{X}$ defined as the composition
    $Q(X) \rto[\psi_{X}] Q(F(X)) \rto[Q(\gamma)] Q(X)$.  Then the
    object $(\mu.F,\chi,\mu.\phi_{\mu.F})$ of
    $\int \Qmu \simeq \alg_{\iQ}(\lF)$ is initial. 

  \end{enumerate}
\end{theoremEnd}
Notice that, in the second item of the above statement, the maps
$\phi_{X}$ need not be \sp, since we do not require the maps
$\psi_{X}$ to be \sp. Therefore, the least fixed point of
$\phi_{\mu.F}$ might not coincide with the least element of
$Q(\mu.F)$.

\begin{proofEnd} For the first statement,
  consider that the structure $\chi : \nu.F \rto
  F(\nu.F)$ is invertible, so elements in
  $\Qnu(\nu.F,\chi)$ are exaclty post-fixed points of the map $\phi$.
  The statement is then an obvious application of the sufficient
  condition given for the extisence of a \final object in $\iQ$
  \newline
  The second statement is also an application of the sufficient
  condition given for the existence of an initial object in
  $\iQ$.  For this, we need to argue that least fixed
  points of the maps $\phi_{X} : Q(X) \rto
  Q(X)$ are preserved by maps of the form $Q(f)$ with
  $f$ and algebra morphism.  This is a standard consequence in fixed
  poiont theory of commutativity of the following diagram:
  \begin{center}
    \begin{tikzcd}[ampersand replacement=\&]
      \ar[bend left=20]{rr}{\phi_{X}}
      Q(X) \ar{d}{Q(f)} \ar{r}{\psi_{X}} \& Q(F(X)) \ar{d}{Q(F(f))}
      \ar{r}{Q(\gamma)} \& Q(X) \ar{d}{Q(f)} \\
      \ar[bend right=20,swap]{rr}{\phi_{Y}}
      Q(Y) \ar{r}{\psi_{Y}} \& Q(F(Y)) \ar{r}{Q(\delta)} \& Q(Y)
    \end{tikzcd}
  \end{center}
  and of the fact that $Q(f)$ has a \ra.
\end{proofEnd}

\section{Examples}
\label{sec:examples}

\begin{Example}[Categories of poset-valued sets]
  \label{ex:DePaivaSchalk}
  A main reason for developing this research was to understand the
  structure of categories of the form \QSet introduced in
  \cite{SchalkDePaiva2004} and, possibly, to devise generalization of
  this kind of constructions.

  For $Q = (Q,e,\qmult)$ a unital commutative quantale, the category
  $\QSet$ is defined as follows:
  \begin{itemize}
  \item An object is a pair $(X,\alpha)$ with $X$ a set and
    $\alpha :X \rto Q$ a function.
  \item An arrow $R : (X,\alpha) \rto (Y,\beta)$ is a relation
    $R : X \rto Y$ such that, for all $x \in X$ and $y\in Y$, $xRy$
    implies $\alpha(x) \leq \beta(y)$.
  \end{itemize}
  Let $Q(X) \eqdef Q^{X}$ be the set of functions from $X$ to $Q$,
  pointwise ordered and, 
  for a relation $R : X \rto Y$, let $Q(R) : Q^{X} \rto Q^{Y}$ be
  defined by
  \begin{align*}
    Q(R)(\alpha)(y) & \eqdef \bigvee_{x R y} \alpha(x)\,,\quad
    \text{for each $\alpha \in Q^{X}$ and $y \in Y$.} 
  \end{align*}
  It is easily verified that $Q(X)$ is a complete lattice, that $Q(R)$
  is \sp, and that $Q$ is a functor $\Rel \rto \Pos$ factoring through
  $\SLatt$. Clearly, $\QSet$ arises as the total category $\iQ$.
\begin{commentout}
    For example:
    \begin{align*}
      Q(S)(Q(R)(\alpha)(y)) & = Q(S)(\bigvee_{x R y} \alpha(x)) =
      \bigvee_{ySz} \bigvee_{x R y} \alpha(x) = \bigvee_{xR;Sz}
      \alpha(x) = Q(R;S)\,.
    \end{align*}
    It also seen that $Q$ is actuallty a 2-functor. (So, for example,
    $Q(\pi_{X}^{\ast})$ is a right adjoint in $\SLatt$, which means
    that it also preserves infima).
\end{commentout}
While, for each set $X$, $Q(X)$ is a quantale (with the product
structure), for a relation $R : X \rto Y$, only the inclusions
$Q(R)(1) \leq 1$ and
$Q(R)(\alpha \qmult \beta) \leq Q(R)(\alpha) \qmult Q(R)(\beta)$
hold---that is, $Q(R)$ is comonoidal.  Nonetheless, and recalling that
$Q(\singleton)$ can be identified with $Q$, if we define
\begin{align*}
  u & \eqdef e \in Q(\singleton)\,, & \mu_{X,Y}(\alpha,\beta)(x,y)
  \eqdef \alpha(x) \qmult \beta(y)\,,
\end{align*}
then $\mu_{X,Y}$ is natural (and not merely lax-natural), $u$ and
$\mu$ make $Q$ into a monoidal functor, and $\mu_{X,Y}$ is
bilinear. Thus $Q$ \emph{monoidally} factors though $\SLatt$.  The
monoidal structure of $\QSet$ exhibited in \cite{SchalkDePaiva2004} is
the one induced by $\mu_{X,Y}$. Recall that $\Rel$ is, w.r.t. the
monoidal structure given by $\times$ and $\singleton$, compact closed
category, that is, we have $X \impl Y = X \times Y$. It is verified
that
\begin{align}
  \iota_{X,Y}(\alpha,\beta)(x,y) & = \alpha(x) \impl \beta(y)\,, \quad
  \text{thus, for $\omega \in Q$}\quad
  \omega_{X}(\alpha)(x) = \alpha(x) \impl \omega\,.
  \label{eq:structureRel}
\end{align}

\begin{theoremEnd}[category=examples, all end]{lemma}
  \label{lemma:comonoidal}
  For each $R : X \rto Y$, $Q(R)$ is comonoidal. That is,
  $Q(R)(1) \leq 1$ and
  $Q(R)(\alpha \cdot \beta) \leq Q(R)(\alpha) \cdot Q(R)(\beta)$, for
  each $\alpha,\beta \in Q(X)$.  If, for each $x,x',w$ such that $xRw$
  and $x'Rw$, there exists $z$ such that $z R w$,
  $\alpha(x) \leq \alpha(z)$, and $\beta(x') \leq \beta(z)$, then
  $Q(R)(\alpha \cdot \beta) = Q(R)(\alpha \cdot \beta)$.
\end{theoremEnd}
\begin{proofEnd}
  We have
  \begin{align*}
    Q(R)(\alpha\cdot \beta)(w) & = \bigvee_{x R w} \alpha(x)\cdot
    \beta(x) 
    \leq^{\ast}  \bigvee_{x R w,x'Rw} \alpha(x)\cdot \beta(x') \\
    &
    = (Q(R)(\alpha)(w)) \cdot (Q(R)(\beta)(w)) = Q(R)(\alpha \cdot \beta)(w)\,.
  \end{align*}
  Next, if for each $x,x',w$, $xRw$, $x'Rw$, implies
  $\alpha(x) \leq \alpha(z)$ and $\beta(x') \leq \beta(z)$, for some
  $z$ such that $z R w$, then the set
  $\set{\alpha(x)\cdot \beta(x)\mid xRw}$ is cofinal in
  $\set{\alpha(x)\cdot \beta(x')\mid xRw, x'Rw}$ and the inequality
  marked above is an equality.
\end{proofEnd}

\begin{theoremEnd}[category=examples, all end]{proposition}
  $\mu_{X,Y}$ is a natural transformation and, together with $u$,
  makes $Q$ into a monoidal functor.
\end{theoremEnd}
\begin{proofEnd}
  While a direct verification is immediate, we use
  Lemma~\ref{lemma:comonoidal}.
  \newline
  Notice
  that $\mu_{X,Y}(\alpha,\beta) = \talpha \cdot \tbeta$, where for
  each $(x,y) \in X \times Y$, $\talpha(x,y) = \alpha(x)$ and where
  $\tbeta$ is defined similarly.
  \newline
  For $R : X \rto X'$ and
  $S : Y \rto Y'$, the first part of the Lemma yields the
  following inclusion:
  \begin{align*}
    Q(R \times S)(\mu(\alpha,\beta)) & = Q(R \times Y)(\talpha \cdot
    \tbeta)
    \leq Q(R \times S)(\talpha) \cdot Q(R \times S)(\tbeta)
    \\
    & =
    \wt{Q(R)(\alpha)} \cdot \wt{Q(S)(\beta)} = \mu_{X',Y'}(Q(R)(\alpha), Q(S)(\beta))\,.
  \end{align*}
  For the reverse inclusion we use the second part of the Lemma.
  Let therefore $x,x',y,y',w,u$ be such that
  $(x,y) R \times S (w,u)$ and $(x',y') R'\times S' (w,u)$. Then
  $(x,y') R \times S (w,u)$ and
  $\talpha(x,y') = \alpha(x) = \talpha(x,y)$ and, similarly,
  $\tbeta(x,y') = \beta(y') = \tbeta'(x',y')$.  Thus, the inequality
  in the display above is an equality.
\end{proofEnd}

\begin{commentout}
  As an exercise, we compute $\iota_{X,Y}$ using the formula given
  in Theorem~\ref{thm:iotaSLatt}.  Observe that
  \begin{align*}
    Q(\ev_{X,Z})(\mu_{X,X \times Z}(\alpha,\beta))(z) & = \bigvee_{x
      \in X} \alpha(x) \cdot \beta(x,z)\,,
  \end{align*}
  and therefore
  \begin{align*}
    (\iota_{X,Z}(\alpha,\gamma))(x,z) & = \bigvee \set{\beta \in
      Q(X\times Z)\mid Q(\ev_{X,Z})(\mu_{X,X \times
        Z}(\alpha,\beta))
      \leq \gamma}(x,z) \\
    & = \bigvee \set{ \beta(x,z) \in Q \mid Q(\ev_{X,Z})(\mu_{X,X
        \times
        Z}(\alpha,\beta)) \leq \gamma}\\
    & = \bigvee \set{\beta(x,z) \mid \forall x',z'
      \;\;\alpha(x')\beta(x',z') \leq \gamma(z')} = \alpha(x) \impl
    \gamma(z)
  \end{align*}
  as expected from \cite{SchalkDePaiva2004}.
  \begin{proposition}
    A dualizing object of $\QSet$ is exaclty of the form
    $(1,\omega)$, where $1 = \singleton$ is a singleton and $\omega$
    a function tha picks a dualizing element of $Q$.
  \end{proposition}
  \begin{proof}
    A dualizing element in $\Rel$ is necessarily a singleton, say
    $\singleton$. The canonical isomorphism $j_{X} : X \rto \SSX$ is
    the relation
    $J_{X} = \set{(x,((x,\singlel),\singlel)) \mid x \in X}
    \subseteq X \times ((X \times \singleton) \times \singleton)$ .
    Thus, for $(\singleton,\omega)$ to be dualizing the sufficient
    and necessary condition is that
    \begin{align*}
      \alpha(x) & = Q(J_{X})(\alpha)((x,\singlel),\singlel) =
      \iota_{\SX,0}(\iota_{X,0}(\alpha,\omega),\omega)((x,\singlel),\singlel)
      = (\alpha(x) \impl \omega(\singlel)) \impl \omega(\singlel)\,,
    \end{align*}
    for each $X$, $\alpha \in Q(X)$, and $x \in X$.  Letting
    $X = \singleton$, then we see that $\omega$ is a dualizing
    element of $Q$ and, of course, being a dualizing element of $Q$
    is enough for the relation above to hold.
  \end{proof}
  \begin{proposition}
    A nuclear object of $\QSet$ is exactly of the form $(X,\alpha)$
    where the image of $\alpha$ takes values in the group of
    invertible elements elements of $Q$.
  \end{proposition}
  \begin{proof}
    In \Rel, the canonical arrow $X^{\ast} \tensor X \rto X \impl X$
    is simply the identity of $X \times X$. Thus, to have this
    isomorphism to lift to $\QSet$, we need to ask the equality
    \begin{align*}
      (\alpha(x) \impl 1) \cdot \alpha(y) & = \alpha(x) \impl
      \alpha(y)\,,
    \end{align*}
    in particular, when $x = y$.  It is easily verified that
    $ (\alpha \impl 1)\cdot \alpha = \alpha \impl \alpha$ holds if
    and only if $\alpha$ is an invertible element of $Q$ and that,
    when $\alpha$ have this property,
    $(\alpha \impl 1) \cdot \beta = \alpha \impl \beta$ holds for
    each $\beta \in Q$.
  \end{proof}
\end{commentout}

In \cite{SchalkDePaiva2004} this construction was also generalized as
follows. Let $F : \Rel \rto \Rel$ be a 2-functor (equivalently, the
relation lifiting of weak-pullback preserving endofunctor of $\Set$,
see \cite{KurzVelebil2016}) that is comonoidal and whose natural
arrows $\nu_{X,Y} : F(X \times Y) \rto F(X) \times F(Y)$ are
functional.  An object of $\QFSet$ is a pair $(X,\alpha)$ with
$\alpha \in Q(F(X)) ( =Q^{F(X)})$. An arrow in $\QSet$ from
$(X,\alpha)$ to $(Y,\beta)$ is a relation $R : X \rto Y$ such that,
for each $x \in F(X)$ and $y \in F(Y)$, if $x F(R) y$, then
$\alpha(x) \leq \beta(y)$.
Let $\conv : \Rel \rto \Rel^{op}$ be the converse functor, identity on
objects and picking the converse of a given relation. Given $F$ as
above, $\nu_{X,Y}^{\conv} : F(X) \times F(Y) \rto F(X \times Y) $
makes $F$ into a monoidal functor.  Recall next that if
$F : \Rel \rto \Rel$ and $Q : \Rel \rto \Pos$ are monoidal, then their
composition $Q \circ F$ is monoidal as well, with natural
transformation
\begin{align*}
  \mu^{Q \circ F}_{X,Y} &  = Q(\mu_{X,Y}^{F}) \circ \mu^{Q}_{F(X),F(Y)}
  : Q(F(X)) \times Q(F(Y)) 
  \rto Q(F(X \times Y))\,.
\end{align*}
Clealry $Q \circ F$ monoidally factors through $\SLatt$ and the
equality $\QFSet = \int Q \circ F$ holds. The monoidal structure
described in \cite{SchalkDePaiva2004} is the one arising from the
monoidal $\mu^{Q \circ F}$.

\textEnd[category=examples]{
If $Q(X) = Q^{X}$, then we have a monoidal
closed structure on $\int Q \circ F$, which can briefly described as
follows:
\begin{align}
  \nonumber
  (X,\alpha) \tensor (Y,\beta) & = (X \tensor Y, \gamma) \,, \qquad\gamma
  \eqdef Q(\mu^{F}_{X,Y})( \mu^{Q}_{FX,FY}(\alpha,\beta)) \in
  Q(F(X \times Y))\,,
  \intertext{or, more explicitly, for
    $z \in F(X \times Y)$,} 
  \gamma(z) & = \bigvee
  \set{\alpha(x)\beta(y) \mid (x,y) \mu^{F}_{X,Y} z}\,.
  \label{eq:generalQF}
\end{align}
In \cite{SchalkDePaiva2004} the category \QSet was generalized as
follows. Let $F : \Rel \rto \Rel$ be a 2-functor (that is, the
relation lifiting of weak-pullback preserving endofunctor) which is
comonoidal with natural transformations
$\nu_{X,Y} : F(X \times Y) \rto F(X) \times F(Y)$.  An object of
$\QFSet$ is a pair $(X,\alpha)$ with
$\alpha \in Q(F(X)) ( =Q^{F(X)})$. An arrow in $\QSet$ from
$(X,\alpha)$ to $(Y,\beta)$ is a relation $R : X \rto Y$ such that,
for each $x \in F(X)$ and $y \in F(Y)$ such that $x F(R) y$,
$\alpha(x) \leq \beta(y)$. That is, we have the identity
$\QFSet = \int Q \circ F$.
}
\begin{theoremEnd}[category=examples, all end]{lemma}
  \label{lemma:PFsets}
  For an endofunctor $F$ of $\Rel$ with the above properties,
  $\nu_{X,Y}^{\conv} : F(X) \times F(Y) \rto F(X \times Y) $ makes
  $F$ into a monoidal functor.
\end{theoremEnd}
\begin{proofEnd}
  Since $F$ is a 2-functor, then it preserves left and right
  adjoints in $\Rel$, and then it commutes with the converse
  functor. Indeed, writing $R : X \rto Y$ as the composition
  $X \rto[\pi_{X}^{\conv}] R \rto[\pi_{Y}] Y$, we have
  $R^{\conv} = \pi_{Y}^{\conv} \cdot \pi_{X}$. Then
  $F(R^{\conv}) = F(\pi_{Y}^{\conv}) \cdot F(\pi_{X}) =
  F(\pi_{Y})^{\conv} \cdot F(\pi_{X}^{\conv})^{\conv} =
  (F(\pi_{X}^{\conv} \cdot \pi_{Y}))^{\conv} = F(R^{\conv})$.  Then
  commutativity of the diagram on the left immediately follows from
  commutativity of the one on the right:
  \begin{align*}
    \begin{tikzcd}[ampersand replacement=\&]
      F(X' \times Y') \& \ar[swap]{l}{\nu_{X',Y'}^{\conv}} F(X') \times F(Y') \\
      F(X \times Y) \ar{u}{F(R \times S)}\& \ar[swap]{l}{\nu_{X,Y}^{\conv}}
      \ar[swap]{u}{F(R) \times F(S)}F(X) \times F(Y) 
    \end{tikzcd}
    \qquad
    \begin{tikzcd}[ampersand replacement=\&]
      F(X' \times Y') \ar{r}{\nu_{X',Y'}} \ar[swap]{d}{F(R^{\conv} \times
        S^{\conv})} \&  F(X') \times F(Y')
      \ar{d}{F(R^{\conv}) \times F(S^{\conv})}
      \\
      F(X \times Y) \ar{r}{\nu_{X,Y}} \& 
      F(X) \times F(Y)  
      \end{tikzcd}
    \end{align*}
  \end{proofEnd}
  \textEnd[category=examples]{
  It is required in \cite{SchalkDePaiva2004} that the relations
  $\nu_{X,Y}$ are total functions.  Notice that if $R$ is the converse of a total
  function $f$, then
  $Q(R)(\alpha)(y) =  \bigvee_{x R y} \alpha(x) = \bigvee_{f(y) = x} \alpha(x) =
  \alpha(f(y))$.
  Thus, if $\mu_{X,Y} = \nu_{X,Y}^{\conv}$ with a functional
  $\nu_{X,Y} = \langle\nu_{1},\nu_{2}\rangle : F(X \times Y) \rto F(X)
  \times F(Y)$, then the general formula in \eqref{eq:generalQF}
  applied to $Q \circ F$ yields
  \begin{align*}
    (X,\alpha) \tensor (Y,\beta) & = (X \tensor Y, \gamma)
    \qquad\text{with}\qquad \gamma(z) =
    Q(\nu_{X,Y}^{c}(z))(\mu^{Q}_{X,Y}(\alpha,\beta))(z) =
    \alpha(\nu_{1}(z))\beta(\nu_{2}(z)) \,.
  \end{align*}
  That is, the tensor product of $\QFSet$ described in
  \cite{SchalkDePaiva2004} is the canonical one of $\int Q \circ F$.
}

\medskip

A wide generalization of the category $\QSet$ arises if we replace
$\Rel$ with $\QRel$, the category of relations or matrices over $Q$.
An object of the category of $\QRel$
is a set, while an arrow $X \rto Y$ is a map
$\psi : X \times Y \rto Q$. Identity is the Kronecker $\delta$, while
composition of $\psi: X \rto Y$ and $\chi : Y \rto Z$ is defined as
follows:
\begin{align*}
  (\psi \cdot \chi)(x,z) & \eqdef \bigvee_{y \in Y} \psi(x,y)\qmult
  \chi(y,z)\,.
\end{align*}
The category $\QRel$ is compact closed
\cite{KellyLaplaza1980,Rosenthal1994b}, with
$X \tensor Y = X \impl Y = X \times Y$, $I = 0 = \singleton$, and
$\SX = X$ on objects.  $\QRel$ is also a quantaloid, that is, a
category enriched over $\SLatt$. From this it follows that
$\QRel(\singleton,\funNB)$ is a monoidal functor from $\QRel$ to
$\SLatt$. Thus, $\int \QRel(\singleton,\fun)$ is monoidal
closed. Since formulas as in \eqref{eq:structureRel} still hold in
this wider setting, then dualizing objects of
$\int \QRel(\singleton,\funNB)$ correspond to dualizing elements in
$\QRel(\singleton,\singleton) \simeq Q$.

Recall that a fuzzy set is a map $\alpha : X \rto\relax [0,1]$.  The
unit interval $[0,1]$ has several quantale structures, for example, it
is a complete Heyting algebra. A \emph{norm} on $[0,1]$ is indeed a
continuous multiplication making $[0,1]$ into a quantale. Fuzzy sets
have been organised in several categories of fuzzy relations, see e.g.
\cite{Winter2007,HWW2014,Mockor2015}.  For $Q = [0,1]$ (with a norm),
the category $\int \QRel(\singleton,\funNB)$ accommodates most of
these constructions.

\end{Example}

\begin{Example}[Non-uniform totality spaces]
	\label{exp:nuts}
  Let us recall some definitions from \cite[\S 1.2.2]{Farzad2023}.
  A \emph{totality space} is a pair $(X,\alpha)$ with $\alpha$ is an
  upward closed (w.r.t. in containement) collection of susbets $X$.
  For a relation $R \subseteq X \times Y$ and $A \subseteq X$, let us
  define
  $R^{\sharp}(A) \eqdef \set{y \in Y \mid \exists x \in A \Tst xRy}$.
  A \emph{morphism of totality spaces} from $(X,\alpha)$ to $(Y,\beta)$ is a
  relation $R \subseteq X \times Y$ such that, for each
  $A \in \alpha$, $R^{\sharp}(A) \in \beta$, see  \cite[Lemma 33]{Farzad2023}.

  Let $\UP(X)$ be the set of upsets of $\P(X)$ (so, in particular,
  $\UP(X)$ is a complete lattice). We extend $\UP$ to a functor from
  $\Rel$ to $\SLatt$ as follows: for $R \subseteq X \times Y$ and
  $\alpha \in \UP(X)$, we let
  \begin{align}
    \label{eq:defFunctor}
    \UP(R)(\alpha) & \eqdef \;\;\uparrow \exists_{R^{\sharp}}(\alpha)\,,
  \end{align}
  that is, $\UP(R)$ is the upward closure of the direct image of
  $R^{\sharp}$.  Notice that
  \begin{align}
    \label{eq:UPRleftadj}
    \uparrow \exists_{R^{\sharp}}(\alpha) \subseteq \beta & \Tiff
    \alpha \subseteq (R^{\sharp})^{-1}(\beta)\,,
  \end{align}
  so $\UP(R)$, having a right adjoint, is \sp.  
  Let us consider the category $\int \UP$. It is easily seen that the
  equivalent conditions in \eqref{eq:UPRleftadj} amount to the
  statement that $R^{\sharp}(A) \in \beta$, for each $A \in
  \alpha$. That is, $\Nuts$ arises as the total category of the
  functor $\UP$.  
  The monoidal structure of $\Nuts = \int \UP$ is obtained by means of
  \begin{align*}
    \set{\singleton} & \in  \UP(\singleton)\, \Tand \mu_{X,Y}(\alpha,\beta) = \,\uparrow \set{a \times b \mid a \in\alpha, \,b \in \beta }\,,
  \end{align*}
  where the latter is natural and bilinear. 

  \begin{theoremEnd}[category=examples, all end]{lemma}
    The following holds:
    \begin{align*}
      \iota_{Y,Z}(\beta,\gamma) & = \set{R \subseteq Y \times Z \mid
        R^{\sharp}(B) \in \gamma, \text{ for all $B \in \beta$ }}\,.
    \end{align*}
    In particular, $I = (\singleton,\set{\singleton})$ is a dualizing
    element of \Nuts.
  \end{theoremEnd}
  \begin{proofEnd}
    Notice that, for $B \subseteq Y$ and $R \subseteq Y \times Z$.
    \begin{align*}
      \ev^{\sharp}(B \times R) & = \set{c \in Z \mid \exists b,b',c'
        \Tst b \in B, (b',c') \in R, (b,b',c') \ev c} \\
      & = \set{c \in Z \mid \exists b,b',c'
        \Tst b \in B, (b',c') \in R, b = b' \Tand c' = c} \\
      & = \set{c \in Z \mid \exists b
        \Tst b \in B \Tand bRc} = R^{\sharp}(B)\,.
    \end{align*}
    Next:
    \begin{align*}
      \exists_{\ev^{\sharp}}(\upset \set{B \times R \mid B \in \beta})
      \subseteq \gamma 
      \Tiff & \forall B \in \beta \;\; \ev^{\sharp}(B \times R) \in \gamma \\
      \Tiff & \forall B \in \beta \;\; R^{\sharp}(B) \in \gamma \,.
    \end{align*}
    Therefore:
    \begin{align*}
      R \in \iota_{Y,Z}(\beta,\gamma) & = \bigcup \set{\alpha \in
        \UP(Y\times Z) \mid \UP(\ev)(\mu(\beta,\alpha)) \subseteq
        \gamma)}  \\
      & \Tiff \UP(\ev)(\mu(\beta,\uparrow R)) \subseteq \gamma \\
      & \Tiff R^\sharp(B) \in \gamma, \forall B \in \beta\,.
    \end{align*}
    Let now $\omega = \set{\singleton}$,  then
    \begin{align*}
      \iota_{X,\singleton}(\alpha,\omega) & = \set{R \subseteq X
        \times \singleton \mid R^{\sharp}(A) = \singleton, \;\forall A
        \in \alpha}\\
      & \simeq \set{A' \subseteq X
        \mid A' \cap A \neq \emptyset, \;\forall A
        \in \alpha} \defeq \alpha^{\perp}\,.
    \end{align*}
    Thus, setting $(X,\alpha)^{\perp} \eqdef (X,\alpha^{\perp})$ and
    considering that $\alpha^{\perp}{}^{\perp} = \alpha$, then we see
    that $(X,\alpha) \impl (\singleton,\omega)$ is canonically
    isomorphic to $(X,\alpha)^{\perp}$, which is a duality of \Nuts,
    we deduce that $(\singleton,\omega)$ is dualizing.
  \end{proofEnd}
\end{Example}

\begin{example}[Focused orthogonality categories]
  \label{example:focs}
  The category \Nuts is an example of \saut category $\OCorth$ arising
  from a focused orthogonality structure \cite{HS2003}. It was
  recently recognised \cite{Fiore2023} that these categories arise as
  total categories of some functor. Let us reconsider these categories
  with the help of the theory we have developed.

  For a monoidal category $\C$, the hom functor
  $\C(I,\funNB) : \C \rto \Set$ is monoidal, since we can associate to
  $x : I \rto X$ and $y : I \rto Y$, the arrow
  $x \bullet y \eqdef I \rto[\simeq] I \tensor I \rto[x \tensor y] X
  \tensor Y$. The powerset functor $P : \Set \rto \SLatt$ is also
  monoidal. By composing these functors, we obtain the functor
  $P(\C(I,\funNB)) : \C \rto \SLatt$ whose  monoidal
  structure is given by
  \begin{align*}
    \mu_{X,Y}(\alpha,\beta) = \set{x\bullet y \mid x \in \alpha, \,y
      \in \beta}\,. 
  \end{align*}
  Next, if $\C$ is \saut, then there is a natural isomorphism
  $\sharp : \C(X,0) \rto \C(I,\SX)$ with the property that, for
  $x : I \rto X$ and $y : X \rto 0$,
  $y \circ x = \ev_{X,0} \circ (x \bullet y^{\sharp})$. Let now
  $\independent \subseteq \C(I,0)$. Using the notation of
  \cite{Fiore2023}, for $\alpha \subseteq \C(I,X)$, we have
  \begin{align*}
    P(\sharp)(\alpha^{\perp}) & = \set{y^{\sharp}
      \mid y \in \C(0,X),\, y \circ
        x \in \independent,
        \text{ for all } x \in \alpha} \\
      &
     = \set{y \in \C(I,\SX) \mid \ev_{X,0} \circ (x \bullet y)\in \independent,
       \text{ for all } x \in \alpha} \\
     & = \bigcup \set{ y \in \C(I,\SX) \mid P(\C(I,\ev_{X,0}))(\mu_{X,\SX}(\alpha,\set{y}))
       \subseteq \independent} \\
     & = \bigcup \set{ \beta \subseteq \C(I,\SX) \mid
       P(\C(I,\ev_{X,0}))(\mu_{X,\SX}(\alpha,\beta))
       \subseteq \independent} =
     \iota_{X,0}(\alpha,\independent)\,.
   \end{align*}
   That is, the Galois connection described in \cite{HS2003,Fiore2023}
   is, up to the isomorphism $P(\sharp) : P(C(X,0)) \simeq P(\C(I,\SX))$,
   the one we describe in Section~\ref{sec:doubleNegation}. The
   category $\OCorth$ is, therefore, a special case of the given
   construction $\int \Qj$, where $Q$ has been instantiated with
   $P(C(I,\funNB))$ and $\omega = \independent \in P(C(I,0))$.
\end{example}

\section{Conclusions}
\label{sec:conclusions}

We have given exact conditions that allow to lift a functor from $\C$
to $\int Q$. In particular, lifting the monoidal structure of $\C$ is
almost
equivalent to $Q$ being a monoidal functor, as naturality of the
arrows $\mu_{X,Y} : Q(X) \times Q(Y) \rto Q(X \tesnor Y)$ is replaced
by lax naturality.
Similarly, lifting the internal hom to $\int Q$ was argued to be
equivalent to giving a lax extranatural transformation with components
$\iota_{X,Y} : Q(X)^{\op} \times Q(Y) \rto Q(X \impl Y)$ making the
unit and counit of the adjunction in $\C$ maps in $\iQ$. These
conditions were found to be equivalent to the fact that the maps
$\brack{\alpha,\funNB}_{X,Y}$ have a right adjoint and that $\mu$ may
be defined in terms of these maps. 
We argued that these conditions are automatically satisfied
when $Q$ is actually a monoidal functor into $\SLatt$---thus, in this case, $\iQ$ 
is always \SMC.
We investigated then dualizing objects in $\iQ$. We proved that, when
$Q$ is monoidal, a necessary and sufficient condition for $(0,\omega)$
to be dualizing is that $\C$ is \saut and that the maps
$\iota_{X,0}(-,\omega)$ are invertible.  
One of our most fascinating result was to construct a functor $\Qj$ in
an analogous way of the double negation construction in quantale
theory. This gives another total category $\int \Qj$ where
$(0,\omega)$ is a dualizing object. Moreover, we gave a representation
theorem for monoidal functors $Q$ into $\SLatt$ such that $\iQ$ is
\saut, thus giving a sort of phase semantics of proofs, and a
completeness theorem for this semantics.
We ended by studying lifting of
intial algebras and \final coalgebras, that is, of fixed points of
endofunctors. We gave a useful description of categories of algebras
and coalgebras in terms of other total categories. In particular, we
argued that we can always lift initial algebras and \final coalgebras
when $Q$ is a functor having $\SLatt$ as codomain.

\medskip
We focused in this paper on lifting structures---the \SMC one, being
star-autonomous, initial and terminal coalgebras of functors---that
arise in the semantics of proofs for fragments of linear logic
augmented with least and greatest fixed points. The categorical
machinery developed in Section~\ref{sec:liftingFunctors} and
exemplified all along this paper can be applied for lifting different
kind of categorical structures.  For example, one might wish to lift
products from $\C$. It is then verified that a product is lifted from
$\C$ to $\iQ$ via $\mu$ if and only if
$\mu_{X,Y} : Q(X) \times Q(Y) \rto Q(X \times Y)$ is right adjoint to
the canonical map $Q(X \times Y) \rto Q(X) \times Q(Y)$. This
observation is implicit in the literature, see for example
\cite{Santocanale18,Abramsky2019,EvangelouOostBH23}.  Notice that, in
principle, the methodology applies for all kind of categorical
structures and, due to its generality, can be instantiated to several
concrete categories. Categories of generalized metric spaces, which
can be seen as fibrations over $\Set$, are an example. The problem of
lifting set theoretic functors to these categories is also
investigated in \cite{BKV2019}.

The research developed here also raises many natural mathematical
questions that we aim to answer in a close future. For example, it can
be shown that if $\C$ is \saut with the tensor unit $I$ being
dualizing, and if $(I,\omega)$ is a dualizing object of $\iQ$, then
$\omega$ is a dualizing element of the quantale $Q(I)$. Whether the
converse hold, we do not know yet.  We expect the representation
Theorem~\ref{thm:representation} to yield hints on this problem.
Last
but not least, we have emphasized throughout this paper the analogy of
a monoidal functor $Q : \C \rto \SLatt$ with a quantale. Such a
monoidal functor is naturally enriched over $\SLatt$, and the
convolution product makes monoidal the category $\homm{\C}{\SLatt}$ of
functors enriched over $\SLatt$. 
It is not difficult to observe that if $\C$ is \saut, then so is
$\homm{\C}{\SLatt}$. Monoidal functors into $\SLatt$ are actually
monoids in $\homm{\C}{\SLatt}$. We studied in
\cite{delacroix_et_al:LIPIcs.CSL.2023.18} Frobenius monoids in
arbitrary \sautcats. A natural conjecture is that the structure
described in Sections~\ref{sec:staraut} and \ref{sec:doubleNegation},
making $\iQ$ into a \sautcat, is the one of a Frobenius monoid in
$\homm{\C}{\SLatt}$.

\medskip\noindent
\textbf{Acknowledgment.} 
The authors are thankful to Marcelo Fiore for numerous hints and
pointers to the existing literature.

\bibliographystyle{abbrv}
\bibliography{biblio}

\ifcsl

\newpage
\appendix

\renewcommand{\sectionautorefname}{Section}
\newcommand{\proofsOfSection}[1]{
  \section*{Proofs of statements from \autoref{sec:#1}, \nameref{sec:#1}}
  \printProofs[#1]
  \newpage
}

\newcommand{\additionalOfSection}[1]{
  \section*{Additional material for \autoref{sec:#1}, \nameref{sec:#1}}
  \printProofs[#1]
  \newpage
}

 \proofsOfSection{liftingFunctors}
 \proofsOfSection{closed}
 \proofsOfSection{staraut}
 \proofsOfSection{doubleNegation}
 \proofsOfSection{liftingfixedpoints}
 \additionalOfSection{examples}

\fi

\end{document}